\documentclass[a4paper,11pt]{article}

\usepackage{amsmath,amssymb,amsthm,graphicx}

\newtheorem{definition}{Definition}
\newtheorem{lemma}{Lemma}
\newtheorem{theorem}{Theorem}

\begin{document}

\title{Power of Uninitialized Qubits in Shallow Quantum Circuits}

\author{Yasuhiro Takahashi and Seiichiro Tani\\
NTT Communication Science Laboratories, NTT Corporation\\
{\small\tt \{takahashi.yasuhiro,tani.seiichiro\}@lab.ntt.co.jp}}

\date{}

\maketitle

\begin{abstract}
   We study the computational power of shallow quantum circuits with $O(\log n)$ initialized 
   and $n^{O(1)}$ uninitialized ancillary qubits, where $n$ is the input length and the initial 
   state of the uninitialized ancillary qubits is arbitrary. First, we show that such a 
   circuit can compute any symmetric function on $n$ bits that is classically computable 
   in polynomial time. Then, we regard such a circuit as an oracle 
   and show that a polynomial-time classical algorithm with the oracle can 
   estimate the elements of any unitary matrix corresponding to a constant-depth quantum 
   circuit on $n$ qubits. Since it seems unlikely that these tasks can be done with only 
   $O(\log n)$ initialized ancillary qubits, our results give evidences that adding uninitialized 
   ancillary qubits increases the computational power of shallow quantum circuits with only 
   $O(\log n)$ initialized ancillary qubits. Lastly, to understand the limitations of 
   uninitialized ancillary qubits, we focus on near-logarithmic-depth quantum circuits with them 
   and show the impossibility of computing the parity function on $n$ bits.
   \end{abstract}

\section{Introduction}\label{intro}

\subsection{Background and Main Results}

Much attention has been paid to the computational power of shallow (i.e., polylogarithmic-depth) 
quantum 
circuits~\cite{Cleve,Moore,Green,Fenner,Hoyer,Fang,Bera2,Takahashi-sim,Takahashi-cc,Bravyi}. 
A major purpose of this line of research is to understand the differences 
between shallow quantum and classical circuits. In addition, it is strongly motivated by one 
of the most difficult problems concerning quantum circuit implementation: in current and 
near-future technologies, it would be very difficult to keep quantum coherence 
for a period of time long enough to apply many gates.

In discussing the computational power of shallow quantum circuits, polynomially many 
ancillary qubits initialized to, say, $|0\rangle$ are assumed to be available. The initialized 
ancillary qubits are particularly important for quantum circuits since many quantum operations 
require ancillary qubits to preserve their unitary property and store intermediate results. 
Another implementation problem arises here: it is difficult to prepare a large number of 
qubits that are simultaneously initialized to a certain state. Indeed, this problem has often been 
addressed in the literature~\cite{DiVincenzo,Ladd}. However, most papers concerning 
the problem assume a sufficiently long coherence time. In this paper, we address 
these two problems simultaneously.

A straightforward quantum computation model reflecting a short coherence time and a 
limited number of initialized ancillary qubits would be shallow quantum circuits with $O(\log n)$ 
initialized ancillary qubits, where $n$ is the input length. However, their computational power 
seems quite low since each step of them can utilize only a small number 
of intermediate results. In fact, it is not even known whether such a circuit 
can compute the OR function on $n$ bits, and it seems unlikely that it can. Therefore, 
it is highly desirable to find additional ancillary 
qubits satisfying the following conditions: they should be easier to prepare than 
initialized ancillary qubits and increase the computational power of shallow quantum 
circuits~with only $O(\log n)$ initialized ancillary qubits. An interesting direction is to 
study qubits in the completely mixed state~\cite{Knill}, but it would be better not to assume 
any particular initial state.

We consider polynomially many uninitialized qubits as additional ancillary qubits. 
More concretely, we study shallow quantum circuits with $O(\log n)$ initialized 
and $n^{O(1)}$ uninitialized ancillary qubits, where we assume that no intermediate 
measurements are allowed. The initial state of the uninitialized ancillary qubits is arbitrary and 
thus they are easier to prepare than initialized ancillary qubits, i.e., they 
satisfy the above first condition on additional ancillary qubits. But do they satisfy 
the second condition? Specifically, are shallow quantum circuits with $O(\log n)$ initialized 
and $n^{O(1)}$ uninitialized ancillary qubits more powerful than those without uninitialized 
ancillary qubits? Although uninitialized ancillary qubits are known to be useful for constructing a 
few efficient quantum circuits~\cite{Barenco,Takahashi-factor}, 
a complexity-theoretic analysis of quantum circuits with such ancillary qubits has not 
yet been~done. 

First, to give evidence of an affirmative answer to the question, we consider 
symmetric functions, 
which are Boolean functions whose output depends only on the number of ones in the 
input bits~\cite{Jukna}. Let ${\cal S}_n$ be the class of symmetric functions on $n$ bits 
that are classically computable in polynomial time. For example, ${\cal S}_n$ includes 
the OR function, for which it is not known whether there exists a shallow 
quantum circuit (consisting of one-qubit gates and CNOT gates) with only $O(\log n)$ 
initialized ancillary qubits, and it seems unlikely that it does. However, any 
function in ${\cal S}_n$ can be computed by adding uninitialized ancillary qubits:
\begin{theorem}\label{can}
Any $f_n \in {\cal S}_n$ can be computed by an $O((\log n)^2)$-depth quantum circuit 
with $n$ input qubits, one output qubit, and $O(\log n)$ initialized and $O(n(\log n)^2)$ 
uninitialized~ancillary qubits such that it consists of the gates in the gate set $\cal G$, 
where $\cal G$ consists of a Hadamard gate, a phase-shift gate with 
angle $2\pi c/2^t$ for any integers $t \geq 1$ and $c$, and a CNOT gate.
\end{theorem}
\noindent
Theorem~\ref{can} gives evidence that shallow quantum circuits with 
$O(\log n)$ initialized and $n^{O(1)}$ uninitialized ancillary qubits are more powerful than 
those without uninitialized ancillary qubits in terms of computing symmetric functions. 
The proof of Theorem~\ref{can} immediately 
implies that the depth of the circuit can be decreased to $O(\log n)$ when the circuit is 
allowed to further include unbounded fan-out gates and unbounded Toffoli gates.

Then, to give further evidence of the computational advantage of using uninitialized ancillary 
qubits, we consider a classical algorithm with an oracle that can perform a shallow quantum 
circuit with them. When the oracle receives a bit string $w$, it performs the circuit with 
input qubits initialized to $|w\rangle$ and sends back the classical 
outcome of the measurement on the output qubit. Let $p(n)$ be a polynomial and $C_n$ be a 
constant-depth quantum circuit on $n$ qubits consisting of the gates in $\cal G$. 
The problem, denoted by MAT$(p(n),C_n)$, is to compute a real number $\alpha_x$ such 
that $|\alpha_x - |\langle0^n|C_n|x\rangle|^2| \leq 1/p(n)$ for 
any input $x \in \{0,1\}^n$, where $C_n$ also denotes its matrix representation. It is not 
known whether the problem has a 
polynomial-time classical algorithm, and it seems unlikely that it does~\cite{Ni}, 
even when we use an oracle that can perform a 
shallow quantum circuit with only $O(\log n)$ initialized ancillary qubits. However, 
the problem can be solved by adding uninitialized ancillary qubits:

\begin{theorem}\label{oracle}
For any polynomial $p(n)$ and a constant-depth quantum circuit $C_n$ on $n$ qubits 
consisting of the gates in $\cal G$, {\rm MAT}$(p(n),C_n)$ can be 
solved with probability exponentially (in $n$) close to 1 by a polynomial-time probabilistic 
classical algorithm with an oracle that can perform an $O(\log n)$-depth quantum circuit with 
$2n$ input qubits, one output qubit, and (no initialized and) $n$ uninitialized ancillary qubits 
such that it consists of the gates in $\cal G$.
\end{theorem}
\noindent
As with Theorem~\ref{can}, Theorem~\ref{oracle} gives evidence that shallow quantum 
circuits with $O(\log n)$ initialized and $n^{O(1)}$ uninitialized ancillary qubits are more 
powerful than those without uninitialized ancillary qubits. More concretely, by the proof 
of Theorem~\ref{oracle}, this is evidence that there exists a probability distribution 
on $\{0,1\}$ that can be generated with uninitialized ancillary qubits but cannot without them. 
This is because, otherwise, {\rm MAT}$(p(n),C_n)$ would be solved by using an oracle 
with only $O(\log n)$ initialized ancillary qubits. We give a brief comment on the number of 
input qubits in the circuit performed by the oracle. If the number is large, a classical algorithm 
can send $0^k$ for large $k$ (besides another bit string) to the oracle and the circuit 
can use a part of the input qubits as a large number of initialized ancillary qubits. 
To avoid this, we restrict the number of input qubits to $2n$.

Lastly, to understand the limitations of uninitialized ancillary qubits, for an arbitrary constant 
$0 \leq \delta <1$, we focus on $O((\log n)^{\delta})$-depth quantum circuits with them and 
consider the computability of the parity function on $n$ bits. Since the depth is $o(\log n)$, 
it is easy to show that the parity function cannot be computed by any such circuit consisting of 
the gates in $\cal G$. This is also the case even when the circuit includes additional gates on a 
non-constant number of qubits:

\begin{theorem}\label{cannot}
Let $0 \leq \delta <1$ be an arbitrary constant. Then, the parity 
function on $n$ bits cannot be computed by any $O((\log n)^{\delta})$-depth 
quantum circuit with $n$ input qubits, one output qubit, and $O(\log n)$ initialized 
and $n^{O(1)}$ uninitialized ancillary qubits such that it consists of the gates in $\cal G$, 
unbounded fan-out gates on $(\log n)^{O(1)}$ qubits, and unbounded Toffoli gates.
\end{theorem}
\noindent
Theorem~\ref{cannot} means that $O((\log n)^{\delta})$-depth quantum circuits 
with $O(\log n)$ initialized and $n^{O(1)}$ uninitialized ancillary qubits are {\em not} more 
powerful than those without uninitialized ancillary qubits 
in terms of computing the parity function, even when they include the two types of gates on 
a non-constant number of qubits. Moreover, Theorem~\ref{cannot} 
implies that the circuit in Theorem~\ref{can} is optimal in the following sense. 
As described in the paragraph following Theorem~\ref{can}, the depth of the circuit 
becomes $O(\log n)$ when the circuit uses 
the gates in $\cal G$, unbounded fan-out gates, and unbounded Toffoli gates. As described 
in Section~\ref{tech}, the circuit is based on the computation of the number of ones in 
the input bits and thus can be regarded as a parity circuit. Thus, the circuit cannot 
be significantly improved simultaneously in terms of both the depth and the number of qubits 
on which unbounded fan-out gates act. This is because, otherwise, we would obtain a parity 
circuit that contradicts Theorem~\ref{cannot}.

\subsection{Imposing the Quantum Catalytic Requirement}

Buhrman et al.~\cite{Buhrman} defined a {\em classical} computation with a logarithmic-size 
clean space and a polynomial-size additional space, which they call a catalytic log-space 
computation. The initial state of the additional space 
is arbitrary, and they impose the catalytic requirement that its state has to be returned to 
the initial one at the end of the computation. They showed a surprising result: it appears that 
such a computation is more powerful than that without the additional space. 
The additional space seems like a catalyst in a chemical reaction.

The corresponding catalytic requirement in our quantum setting is that 
the state of uninitialized ancillary qubits has to be returned to the initial one at the end of 
computation. Since the circuit in Theorem~\ref{can} has no error, by the 
standard technique of uncomputation, it is easy to transform the circuit into the one 
that meets the quantum catalytic requirement without increasing the original asymptotic 
complexity. Thus, Theorem~\ref{can} means that uninitialized ancillary qubits seem like a 
catalyst as in the classical setting~\cite{Buhrman}. When shallow quantum circuits have an 
error, it is not easy to transform them into the ones that meet the quantum catalytic 
requirement and the analysis of such circuits is left for future work.

From a practical point of view, it is even better to decrease the number of uninitialized 
ancillary qubits we need to specially prepare in addition to decreasing the number of 
initialized ones. The quantum catalytic 
requirement allows us to do this in some cases. An example is when we use a shallow quantum 
circuit with uninitialized ancillary qubits in a quantum circuit for Shor's 
factoring algorithm~\cite{Takahashi-factor}. The factoring circuit uses two registers and, 
during some operation, all qubits in one register are idle. Thus, when 
we use a shallow quantum circuit for the operation that meets the above requirement, we 
can regard the idle qubits as uninitialized ancillary qubits since the circuit returns their state 
to the initial one. The use of the circuit in this way requires that the computation has to 
be done with only qubits, which matches our quantum computation model. From a 
complexity-theoretic standpoint, it is also interesting to study a quantum computation 
model with an additional classical space~\cite{Watrous-cc}.

\subsection{Overview of Techniques}\label{tech}

We construct two quantum circuits to obtain the circuit for $f_n \in {\cal S}_n$ 
in Theorem~\ref{can}. The first one is an $O((\log n)^2)$-depth OR reduction circuit 
with $O(n(\log n)^2)$ uninitialized ancillary qubits, 
which reduces the computation of the OR function on $n$ bits to that on $m=O(\log n)$ bits. 
Its first part is a modification of the original OR reduction circuit~\cite{Hoyer} and yields a 
state whose phase depends on the uninitialized ancillary qubits but has a convenient 
form to eliminate the dependency. We apply similar circuits repeatedly 
to add an appropriate phase to that of the state, which eliminates any dependency on 
the uninitialized ancillary qubits. The second circuit is 
an $O(m^{2})$-depth one for $g_m$ with $O(m2^{m})$ uninitialized ancillary qubits. Here, 
$g_m$ is a Boolean function on $m$ bits satisfying that $g_m(s)=f_n(x)$ for any $x\in\{0,1\}^n$, 
where $s\in\{0,1\}^m$ is the binary representation of the number of ones in $x$. The circuit 
is based on the Fourier expansion of $g_m$~\cite{Jukna} and the above method for eliminating 
any dependency on the uninitialized ancillary qubits. For any input $x\in\{0,1\}^n$, we 
first compute $s$ using the OR reduction circuit and then compute $g_m(s)=f_n(x)$ using 
the circuit for $g_m$.

The algorithm in Theorem~\ref{oracle} is based on a polynomial-time probabilistic 
classical algorithm for {\rm MAT}$(p(n),C_n)$ with an oracle~\cite{Ni}, where 
the oracle can perform a commuting quantum circuit for the 
Hadamard test~\cite{Cuevas}. Although initialized ancillary qubits 
can be used to parallelize the Hadamard test~\cite{Takahashi-sim}, it has not been known 
whether uninitialized ancillary qubits are useful for this purpose. We show that they 
can be used like initialized ancillary qubits in parallelizing the Hadamard test. We replace 
the commuting quantum circuit with a new circuit with our parallelizing techniques 
using uninitialized ancillary qubits in the algorithm 
for {\rm MAT}$(p(n),C_n)$, which yields the desired algorithm.

We show Theorem~\ref{cannot} by extending the proof of Bera~\cite{Bera2}. 
Our proof is different from the previous one in that it deals with ancillary qubits 
and unbounded fan-out gates. The key to Theorem~\ref{cannot} is to show that, for any 
quantum circuit $C_n$ with $O(\log n)$ initialized and $n^{O(1)}$ uninitialized ancillary 
qubits such that it may include unbounded Toffoli gates, there exists an initial state of the 
uninitialized ancillary qubits such that $C_n$ with the initial state is well approximated 
by $\tilde C_n$ with the same initial state. Here, $\tilde C_n$ is the circuit obtained from 
$C_n$ by removing unbounded Toffoli gates on a large number of qubits. 
Thus, if $C_n$ is a small-depth quantum circuit for the parity function, then $\tilde C_n$ 
computes the same function with high probability. This is impossible since $\tilde C_n$ does not 
have any gate on a large number of qubits and thus its output does not depend on all input qubits.

\section{Preliminaries}

\subsection{Quantum Circuits and Uninitialized Ancillary Qubits}\label{circuit}

A quantum circuit consists of elementary gates, each of which is in the gate set $\cal G$, 
where $\cal G$ consists of a Hadamard gate $H$, a phase-shift gate 
$Z(\theta)$ with angle $\theta$, and a CNOT gate. Here, 
$H=|+\rangle\langle 0| + |-\rangle\langle 1|$ and 
$Z(\theta)=|0\rangle\langle0|+e^{i\theta}|1\rangle\langle 1|$, 
where $|\pm\rangle=(|0\rangle\pm |1\rangle)/\sqrt{2}$ and $\theta = 2\pi c/2^t$ for 
any integers $t \geq 1$ and $c$. We write $Z(\pi)$ and $HZ(\pi)H$ as $Z$ and $X$, 
respectively. In some cases, we use a fan-out gate and a Toffoli gate as elementary gates. 
Let $k\geq 1$ be an integer. A fan-out gate on $k+1$ 
qubits implements the operation defined as 
$|y\rangle\bigotimes_{j=1}^{k}|x_j\rangle \mapsto |y\rangle\bigotimes_{j=1}^{k}|x_j\oplus y\rangle$ 
for any $y,x_j \in \{0,1\}$, where $\oplus$ denotes addition modulo 2. 
The first input qubit is called the control qubit. 
A $k$-controlled Toffoli gate implements the operation on $k+1$ qubits defined as
$\left(\bigotimes_{j=1}^{k}|x_j\rangle\right)|y\rangle \mapsto 
\left(\bigotimes_{j=1}^{k}|x_j\rangle\right)|y \oplus \bigwedge_{j=1}^{k} x_j\rangle$, 
where $\bigwedge$ denotes the logical AND. The first $k$ input qubits are called the control 
qubits and the last input qubit is called the target qubit. These gates with $k=1$ 
are CNOT gates. When it is permitted to apply a fan-out gate and a Toffoli gate on a 
non-constant number of qubits, they are called an unbounded fan-out gate and an 
unbounded Toffoli gate, respectively.

To simplify the descriptions of quantum circuits, 
we use a $k$-controlled $Z(\theta)$ gate for any $\theta$ described above, which will be  
decomposed into elementary gates. The gate 
implements the operation on $k+1$ qubits defined as $\bigotimes_{j=1}^{k+1}|x_j\rangle 
\mapsto e^{i\theta\bigwedge_{j=1}^{k+1}x_j} \bigotimes_{j=1}^{k+1}|x_j\rangle$ 
for any $x_j \in \{0,1\}$. We can choose an arbitrary qubit as the target qubit and 
the other qubits are called the control qubits. 
The inverse of the gate is the $k$-controlled $Z(-\theta)$ gate. 
When it is permitted to apply the gate on a non-constant number of qubits, it is called an 
unbounded $Z(\theta)$ gate.

The complexity measures of a quantum circuit are its size and depth. The size of a 
quantum circuit is the total size of all elementary gates in the circuit, where the 
size of an elementary gate is the number of qubits on which the gate acts. 
To define the depth, we regard the circuit as a set of layers 
$1,\ldots,d$ consisting of elementary gates, where gates in the same layer act on pairwise 
disjoint sets of qubits and any gate in layer $j$ is applied before any gate in layer $j+1$. The 
depth of the circuit is the smallest possible value of $d$~\cite{Fenner}. 

We deal with a uniform family of polynomial-size quantum circuits $\{C_n\}_{n\geq 1}$, 
where no intermediate measurements are allowed. 
The uniformity means that the function $1^n \mapsto \overline{C_n}$ is classically 
computable in polynomial time, where $\overline{C_n}$ is the classical description of $C_n$. 
Each $C_n$ has $n$ input qubits 
and can have one output qubit and $n^{O(1)}$ ancillary qubits that are divided into 
two groups: $p=O(\log n)$ qubits and the remaining $q$ qubits. 
We assume that, for any $x \in \{0,1\}^n$ and $y \in \{0,1\}$, 
we can initialize the input qubits and output qubit to $|x\rangle$ and $|y\rangle$, 
respectively. We can also initialize the $p$ ancillary qubits to $|0\rangle$,
 which we call initialized ancillary qubits, but we cannot 
initialize the $q$ ancillary qubits and do not know their initial state. They 
are called uninitialized ancillary qubits. When $C_n$ has the output qubit, a measurement in 
the $Z$ basis is performed on it at the end of the 
computation. The classical outcome of the measurement, which is 0 or 1, is called the 
output of $C_n$. A symbol denoting a quantum circuit also denotes its matrix representation 
in the computational basis.

\subsection{Computability of Boolean Functions}\label{boolean}

A Boolean function $f_n$ on $n$ bits is a mapping $f_n:\{0,1\}^n\to \{0,1\}$. We define 
its computability by a quantum circuit with uninitialized ancillary qubits as follows:

\begin{definition}
Let $f_n$ be a Boolean function on $n$ bits and $C_n$ be a quantum 
circuit with $n$ input qubits, one output qubit, and $p$ initialized and $q$ 
uninitialized ancillary qubits. The circuit $C_n$ computes $f_n$ 
if, for any $x\in \{0,1\}^n$ and $y\in\{0,1\}$, when the input qubits and output qubit 
are initialized to $|x\rangle$ 
and $|y\rangle$, respectively, the output of $C_n$ is $y\oplus f_n(x)$ with probability 1, 
regardless of the initial state of the $q$ uninitialized ancillary qubits.
\end{definition}

A Boolean function is called symmetric if its output depends only on the number 
of ones in the input bits~\cite{Jukna}. Let ${\cal S}_n$ be the class of symmetric 
functions on $n$ bits that are classically computable in polynomial time. For example, 
${\cal S}_n$ includes the parity function PA$_n$ and the OR function OR$_n$. 
Here, for any $x =x_1\cdots x_n \in \{0,1\}^n$, ${\rm PA}_n(x)=1$ if $|x|$ is odd and 0 
otherwise, where $|x|=\sum_{j=1}^{n}x_j$. Moreover, 
${\rm OR}_n(x)= 1$ if $|x|\geq 1$ and 0 otherwise.

We define the function associated with $f_n \in {\cal S}_n$ as follows:

\begin{definition}\label{associated}
Let $f_n \in {\cal S}_n$. The function associated with $f_n$ is the 
Boolean function $g_m$ on $m=\lceil \log(n+1)\rceil$ bits defined as follows$:$ 
For any $s=s_1\cdots s_m\in\{0,1\}^m$, $g_m(s)=f_n(1^l0^{n-l})$ if 
$l \leq n$ and 0 otherwise, where $l= \sum_{k=1}^{m}s_k2^{k-1}$.
\end{definition}
\noindent
The function $g_m$ is classically computable in time $n^{O(1)}$ and, 
for any $x \in \{0,1\}^n$, if $s=s_1\cdots s_m$ is the binary representation of 
$|x|$, i.e., $|x| =\sum_{k=1}^{m}s_k2^{k-1}$, then $g_m(s) = f_n(x)$. 

We explain the idea of the original OR reduction quantum circuit~\cite{Hoyer}. The circuit 
has $n$ input qubits and $O(n\log n)$ initialized ancillary qubits, and reduces the computation 
of OR$_n$ to that of OR$_m$, where $m=\lceil \log(n+1)\rceil$. When the input state is 
$|x\rangle$ for any $x \in \{0,1\}^n$, the circuit transforms 
the state of $m$ initialized ancillary qubits into the state 
$\bigotimes_{k=1}^{m}|\varphi_k\rangle$, where 
$|\varphi_k\rangle = (|+\rangle+e^{\frac{2\pi i}{2^k}{|x|}}|-\rangle)/\sqrt{2}$ for any 
$1\leq k \leq m$. If $|x|=0$, then $|\varphi_k\rangle=|0\rangle$ for any $1\leq k \leq m$ and thus 
the output state is $|0^m\rangle$. If $|x|\geq 1$, then $|\varphi_{k}\rangle=|1\rangle$ for some 
$1\leq k \leq m$ and thus the output state is orthogonal to $|0^m \rangle$. 
Let $|x|=\sum_{k=1}^{m}s_k2^{k-1}$ for some $s_k\in\{0,1\}$. It is easy to show that 
the state $\bigotimes_{k=1}^{m}|s_k\rangle$ can be obtained by applying 
${\rm QFT}_{2^m}^\dag$ to the state $\bigotimes_{k=1}^{m}H|\varphi_k\rangle$, where 
QFT$_{2^m}^\dag$ is the inverse of the quantum Fourier transform modulo $2^m$.

\section{Shallow Quantum Circuits for Symmetric Functions}\label{parallel}

\subsection{OR Reduction Circuit with Uninitialized Ancillary Qubits}\label{idea}

Let $f_n \in {\cal S}_n$. We compute $f_n$ on input $x \in\{0,1\}^n$ using the following algorithm:
\begin{enumerate}
\item Compute the binary representation $s \in \{0,1\}^m$ of $|x|$, where 
$m=\lceil \log(n+1)\rceil$.

\item Compute $g_m(s) = f_n(x)$, where $g_m$ is the function associated with $f_n$.
\end{enumerate}
To implement Step 1, we construct an OR reduction circuit $Q_n$ with uninitialized 
ancillary qubits. As described above, we can obtain $s$ using 
$Q_n$ (with a layer of $H$ gates) and 
the standard $O(m)$-depth quantum circuit for QFT$_{2^m}^\dag$ with no 
ancillary qubits~\cite{Nielsen}. To implement Step 2, we construct a quantum 
circuit $R_m$ for $g_m$ with uninitialized ancillary~qubits.

\begin{figure}[t]
\centering
\includegraphics[scale=.32]{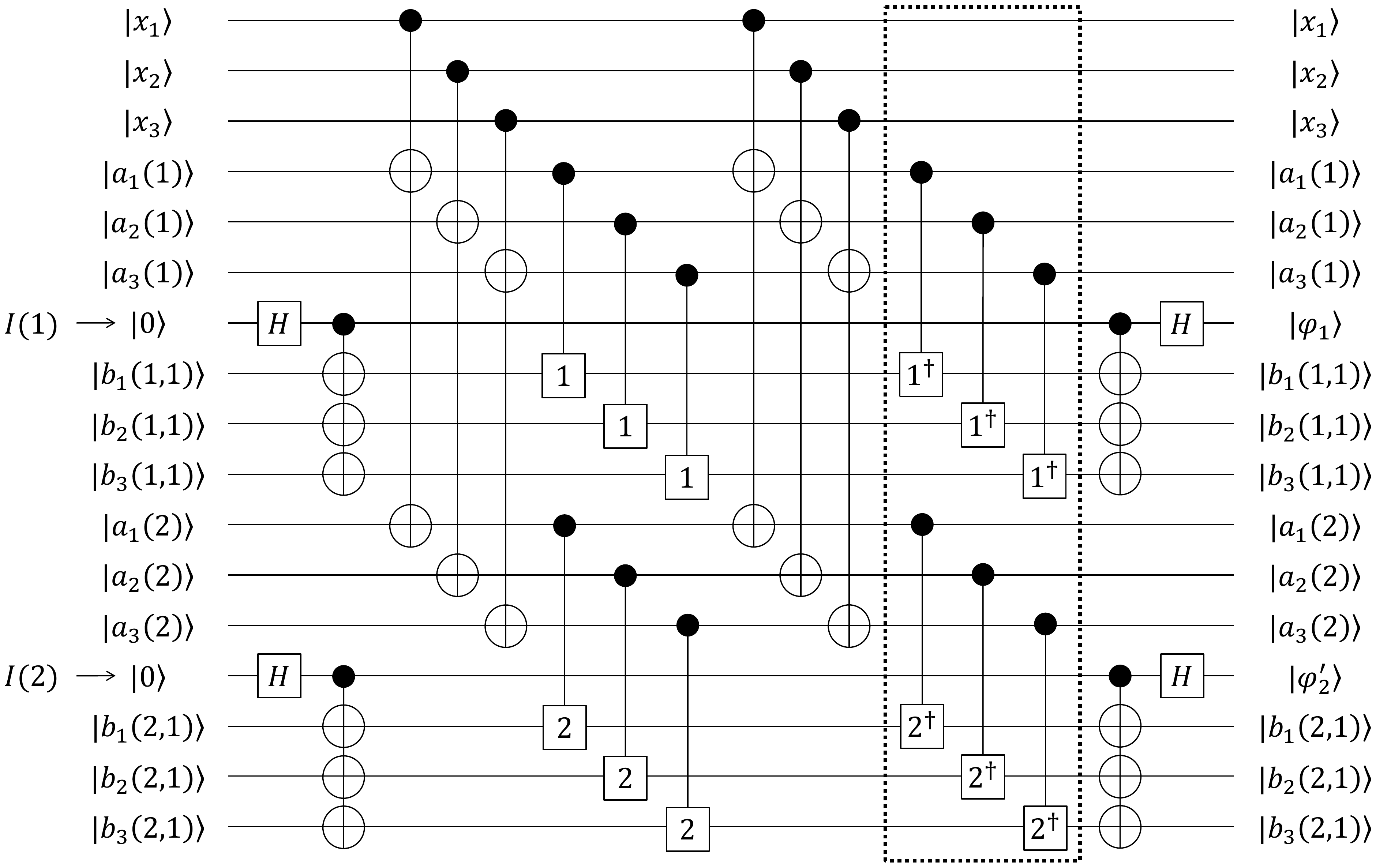}
\caption{The first stage of our OR reduction circuit with input 
$x=x_1x_2x_3\in\{0,1\}^3$. The gate next to the 
$H$ gate is a fan-out gate on four qubits, where the top qubit is the control qubit. 
For any integer $t \geq 1$, the gates $t$ and $t^\dag$ represent a $Z(2\pi/2^t)$ gate and 
its inverse, i.e., a $Z(-2\pi/2^t)$ gate, respectively. The dashed box represents the gates 
added to the original OR reduction circuit.}
\label{figure1}
\end{figure}

The circuit $Q_n$ is an $O((\log n)^2)$-depth OR reduction circuit with $n$ input qubits 
and $O(\log n)$ initialized and $O(n(\log n)^2)$ uninitialized ancillary qubits. 
To explain our idea for constructing $Q_n$, 
we consider the case where $n=3$ (and thus $m=2$). The first stage of $Q_n$ 
is depicted in Fig.~\ref{figure1}, where the initial state of the uninitialized ancillary qubits 
is represented by the (unknown) values $a_j(k),b_j(k,l) \in \{0,1\}$. This circuit is 
obtained by adding the gates in the dashed box to the original OR reduction circuit. 
We want to transform the initial states of the initialized ancillary qubits $I(1)$ 
and $I(2)$ into the states $|\varphi_1\rangle$ and 
$|\varphi_2\rangle$, respectively. If we do not apply the added gates, the 
output state of $I(k)$ is $(|+\rangle + e^{\frac{2\pi i}{2^k}\alpha(k,1)}|-\rangle)/\sqrt{2}$, 
where $\alpha(k,1)=\sum_{j=1}^3 (-1)^{b_j(k,1)}(x_j \oplus a_j(k))$ and $k=1,2$. The phase of 
this state depends on the initial state of the uninitialized ancillary qubits and we 
eliminate the dependency.

The point is that the added gates allow us to obtain an output state of $I(k)$ 
whose phase has a convenient form to eliminate the dependency. More concretely, 
by applying them, the output state of $I(k)$ is 
$(|+\rangle + e^{\frac{2\pi i}{2^k}{\gamma(k,1)}}|-\rangle)/\sqrt{2}$, where 
$\gamma(k,1) = |x| - 2\sum_{j=1}^3 x_j (a_j(k)\oplus b_j(k,1))$. Since 
$e^{\frac{2\pi i}{2}{\gamma(1,1)}} = e^{\frac{2\pi i}{2}|x|}$, the output state of 
$I(1)$ is equal to $|\varphi_1\rangle$ as desired. The dependency is eliminated since the 
terms in $\gamma(1,1)$ other than $|x|$ yield only an angle of a multiple of $2\pi$. 

Unfortunately, the output state of $I(2)$, which is represented as 
$|\varphi_2'\rangle$ in Fig.~\ref{figure1}, is not equal to $|\varphi_2\rangle$ in general since 
the phase $e^{\frac{2\pi i}{2^2}{\gamma(2,1)}}$ depends on the initial states of 
the uninitialized ancillary qubits, where 
$\gamma(2,1) = |x| - 2\sum_{j=1}^3 x_j (a_j(2)\oplus b_j(2,1))$. To eliminate the dependency, 
we consider the second stage where we add an angle $\frac{2\pi}{2^2}\delta(2,2)$ to 
the original angle $\frac{2\pi}{2^2}\gamma(2,1)$ using three new uninitialized ancillary qubits 
(not depicted in Fig.~\ref{figure1}). Here, their initial state 
is $|b_1(2,2)\rangle|b_2(2,2)\rangle|b_3(2,2)\rangle$ for any (unknown) $b_j(2,2) \in \{0,1\}$ 
and $\delta(2,2) = |x|- \gamma(2,1)
-2^2 \sum_{j=1}^3 x_j(a_j(2)\oplus b_j(2,1))(a_j(2)\oplus b_j(2,2))$. The value $\delta(2,2)$ has 
a form similar to $\gamma(2,1)$ and thus we can implement the second stage using a 
quantum circuit similar to the one in Fig.~\ref{figure1}. Since 
$e^{\frac{2\pi i}{2^2}(\gamma(2,1)+\delta(2,2))} = e^{\frac{2\pi i}{2^2}|x|}$, 
we obtain $|\varphi_2\rangle$ as desired.

We generalize the above idea. Let $x=x_1\cdots x_n\in\{0,1\}^n$ be an input. 
We prepare $n$ input qubits $X_1,\ldots,X_n$ and $m$ initialized ancillary 
qubits $I(1),\ldots, I(m)$, where $X_j$ is initialized to $|x_j\rangle$. 
We also prepare $nm(m+3)/2$ uninitialized ancillary qubits, 
which are divided into two groups, $A$ and $B$. Group $A$ consists of $mn$ qubits, 
which are divided into $m$ groups $A(1),\ldots,A(m)$. Each $A(k)$ consists of 
$n$ qubits $A_1(k),\ldots,A_n(k)$, where the initial state of $A_j(k)$ is 
$|a_j(k)\rangle$ for any (unknown) $a_j(k) \in\{0,1\}$. Group $B$ consists of $nm(m+1)/2$ 
qubits, which are divided into $m$ groups $B(1),\ldots,B(m)$. Each $B(k)$ consists of 
$kn$ qubits, which are divided into $k$ groups $B(k,1),\ldots,B(k,k)$. 
Each $B(k,l)$ consists of $n$ qubits $B_1(k,l),\ldots,B_n(k,l)$, where the initial state of 
$B_j(k,l)$ is $|b_j(k,l)\rangle$ for any (unknown) $b_j(k,l)\in\{0,1\}$. The circuit $Q_n$ 
consists of $m$ stages. As an example, Stages 1 and 2 with 
$n=3$ are given in Appendix~\ref{a-correct}. For any $1\leq s \leq m$, 
Stage $s$ is defined as follows:
\begin{enumerate}
\item Apply a $H$ gate to $I(k)$ for every $s \leq k \leq m$ in parallel.

\item Apply a fan-out gate on $n+1$ qubits to $B_1(k,s),\ldots,B_n(k,s)$, and $I(k)$ 
for every $s \leq k \leq m$ in parallel, where $I(k)$ is the control qubit.

\item If $s \geq 2$, then apply a fan-out gate on $s$ qubits to 
$B_j(k,1),\ldots,B_j(k,s-1)$, and $A_j(k)$ for every $s \leq k \leq m$ and $1\leq j \leq n$ 
in parallel, where $A_j(k)$ is the control qubit.

\item Apply a fan-out gate on $m-s+2$ qubits to $A_j(s),A_j(s+1),\ldots,A_j(m)$, and $X_j$ 
for every $1 \leq j \leq n$ in parallel, where $X_j$ is the control qubit.

\item Apply an $s$-controlled $Z(2\pi/2^{k-s+1})$ gate to $B_j(k,s)$ and the following 
qubits for every $s \leq k \leq m$ and $1 \leq j \leq n$ in parallel: 
$A_j(k)$ if $s=1$ and $B_j(k,1),\ldots,B_j(k,s-1)$, and $A_j(k)$ otherwise.

\item Apply the gates in Step 4.

\item Apply the inverse of the gates in Step 5.

\item Apply the gates in Step 3, Step 2, and Step 1 (in this order).
\end{enumerate}

The circuit $Q_n$ outputs the desired state. The 
proof can be found in Appendix~\ref{a-correct}.

\begin{lemma}\label{correct}
Let $x=x_1\cdots x_n \in\{0,1\}^n$ be an input. For any $1 \leq k \leq m$ and 
$1 \leq s \leq k$, the state of $I(k)$ after Stage $s$ is the state
$(|+\rangle + e^{\frac{2\pi i}{2^k}\gamma(k,s)}|-\rangle)/\sqrt{2}$, 
where $\gamma(k,s) = |x| - 2^s \sum_{j=1}^n x_j\bigwedge_{l=1}^{s} (a_j(k)\oplus b_j(k,l))$. 
Moreover, the state of any qubit other than the initialized ancillary qubits is the same as 
its initial one. In particular, the state of $I(k)$ after Stage $k$ is the state $|\varphi_k\rangle$.

\end{lemma}

The circuit $Q_n$ has the desired complexity. The proof can be found in 
Appendix~\ref{a-complexity}.

\begin{lemma}\label{complexity}
The circuit $Q_n$ uses $O(\log n)$ initialized and $O(n(\log n)^2)$ uninitialized ancillary 
qubits, and its depth is $O((\log n)^2)$, when the elementary gate set is $\cal G$.
\end{lemma}

\subsection{Circuit for the Function Associated with a Symmetric Function}

We construct an $O(m^{2})$-depth quantum circuit $R_m$ for $g_m$ with 
$m=\lceil \log(n+1)\rceil$ input qubits, one output qubit, and $O(m2^{m})$ uninitialized 
ancillary qubits, where $g_m$ is the function associated with $f_n \in {\cal S}_n$. 
The circuit uses (a slight modification of) the Fourier expansion of 
$g_m$~\cite{Jukna}: For any $s=s_1\cdots s_m \in \{0,1\}^m$, $g_m(s) = g_m(0^m) +\frac{2}{2^{m}} 
\sum_{t}c_t \bigoplus_{k=1}^m t_ks_k$, 
where $c_t=\sum_{u} g_m(u) (2\bigoplus_{k=1}^m u_kt_k-1)$, $t=t_1\cdots t_m$ ranges 
over $\{0,1\}^m \setminus \{0^m\}$, and $u=u_1\cdots u_m$ ranges 
over $\{0,1\}^m$. Since $m=O(\log n)$ and $g_m$ is classically computable in time 
$n^{O(1)}$, the number of $c_t$'s with $t \in \{0,1\}^m \setminus \{0^m\}$ is $n^{O(1)}$ and the 
function $t \mapsto c_t$ is also classically computable in time $n^{O(1)}$. This implies the 
uniformity of our circuit family for $f_n$.

The circuit $R_m$ with input $s=s_1\cdots s_m \in \{0,1\}^m$ is based on the following algorithm:
\begin{enumerate}
\item Compute the parity value $\bigoplus_{k=1}^mt_ks_k$ for every 
$t\in \{0,1\}^m \setminus \{0^m\}$ in parallel.

\item Prepare $(|+\rangle+e^{\pi i g_m(s)}|-\rangle)/\sqrt{2} = |g_m(s)\rangle$ using the above 
representation of $g_m$.
\end{enumerate}
Since we do not have any initialized ancillary qubit, in Step 1, we can only have the 
parity values on uninitialized ancillary qubits, i.e., $a_t \oplus \bigoplus_{k=1}^mt_ks_k$ for every 
$t\in \{0,1\}^m \setminus \{0^m\}$, 
where the initial state of the uninitialized ancillary qubits is represented by the 
(unknown) values $a_t\in\{0,1\}$. Thus, in Step 2, we have to use such values 
to prepare $(|+\rangle+e^{\pi i g_m(s)}|-\rangle)/\sqrt{2}
=X^{g_m(0^m)} (|+\rangle+e^{\frac{2 \pi i}{2^{m}} \sum_{t}c_t 
\bigoplus_{k=1}^m t_ks_k}|-\rangle)/\sqrt{2}$, which does not depend on $a_t$. The point is 
that this situation is essentially the same as the one where $|\varphi_m\rangle$ is prepared 
by $Q_n$ as described in Section~\ref{idea}, i.e., where we can only have the values 
$a_j(m)\oplus x_j$ for every $1 \leq j \leq n$ and we have to use them to prepare 
$|\varphi_m\rangle = (|+\rangle+e^{\frac{2 \pi i}{2^{m}}{|x|}}|-\rangle)/\sqrt{2}$, 
which does not depend on $a_j(m)$. Thus, roughly speaking, we can construct $R_m$ 
in a similar way to a part of $Q_n$.

A slight difference between these situations is that, in $Q_n$, it is very easy to 
prepare the values  $a_j(m)\oplus x_j$ from the input bits $x_j$, but, in $R_m$, 
we need to consider a quantum circuit for computing the parity values 
$a_t \oplus \bigoplus_{k=1}^mt_ks_k$ from the input bits $s_k$, i.e., for the operation on 
$2^m+m-1$ qubits defined as $|s\rangle\bigotimes_{t}|a_t\rangle \mapsto 
|s\rangle\bigotimes _{t}|a_t \oplus \bigoplus_{k=1}^mt_ks_k\rangle$ for any $s\in\{0,1\}^m$ 
and $a_t\in\{0,1\}$. If we have $m2^{m-1}$ {\em initialized} ancillary qubits, it is easy to construct 
an $O(m)$-depth quantum circuit for the operation using the following algorithm:
\begin{enumerate}
\item Prepare $2^{m-1}$ copies of $s_k$ on the ancillary qubits 
for every $1 \leq k \leq m$ in parallel.

\item Compute the parity value $a_t\oplus \bigoplus_{k=1}^mt_ks_k$ for every $t\in \{0,1\}^m \setminus \{0^m\}$ in parallel.
\end{enumerate}
To implement Step 1, we apply fan-out gates on $2^{m-1}+1$ qubits, each of which 
can be decomposed into an $O(m)$-depth quantum circuit as 
described in Appendix~\ref{a-complexity}. Since it is easy to construct an 
$O(\log m)$-depth quantum circuit for PA$_m$ using a binary 
tree structure, we can implement Step 2 using a parallel application of such circuits. If we 
replace the initialized ancillary qubits with uninitialized ones, the circuit does not work. 
However, applying the circuit again yields the desired values. In fact, the first circuit outputs 
$a_t \oplus \bigoplus_{k=1}^mt_ks_k \oplus d$ for some $d \in \{0,1\}$ that is computed from 
the (unknown) values in  $\{0,1\}$ representing the initial state of the uninitialized ancillary 
qubits, and the second one outputs 
$a_t\oplus \bigoplus_{k=1}^mt_ks_k \oplus d\oplus d = a_t \oplus \bigoplus_{k=1}^mt_ks_k$ as 
desired. Using this circuit, we construct  $R_m$ and show the following lemma. 
The details can be found in Appendix~\ref{a-exp}.

\begin{lemma}\label{exp}
The circuit $R_m$ computes $g_m$. Moreover, it uses no initialized and $O(m2^m)$ 
uninitialized ancillary qubits, and its depth is $O(m^2)$, when the elementary 
gate set is $\cal G$.
\end{lemma}

Combining $R_m$ with $Q_n$ immediately implies Theorem~\ref{can}:

\begin{proof}[Proof of Theorem~\ref{can}]
By Lemmas~\ref{correct}, \ref{complexity}, and~\ref{exp}, we can use $Q_n$ and $R_m$ 
to implement the algorithm for $f_n \in {\cal S}_n$ described at the beginning of  
Section~\ref{idea} and the whole circuit has the desired complexity.
\end{proof}

\section{Classical Algorithms with Access to Shallow Quantum Circuits}\label{solving}

Let $p(n)$ be a polynomial and $C_n$ be a constant-depth quantum circuit on $n$ qubits 
consisting of the gates in $\cal G$. The 
problem MAT$(p(n),C_n)$ is to compute a real number $\alpha_x$ such that 
$|\alpha_x - |\langle0^n|C_n|x\rangle|^2| \leq 1/p(n)$ for any input $x \in \{0,1\}^n$. 
For any $x,w\in\{0,1\}^n$, we define 
$F_n(x,w)=\langle x|C_n^\dag (\bigotimes_{j=1}^n Z_j^{w_j}) C_n|x\rangle$, where 
$w=w_1\cdots w_n$ and $Z_j$ is $Z$ applied to the $j$-th qubit of $C_n$. 
As shown in~\cite{Ni}, {\rm MAT}$(p(n),C_n)$ 
can be solved with probability exponentially (in $n$) close 
to 1 if there exists a probabilistic algorithm $A_{F_n}$ such that, for any $x,w \in \{0,1\}^n$, 
the probability that $|A_{F_n}(x,w)-F_n(x,w)| \leq 0.5/p(n)$ is exponentially close to 1. In fact, 
due to the Chernoff-Hoeffding bound, the algorithm for MAT$(p(n),C_n)$ on 
input $x\in\{0,1\}^n$ is described with some $K=n^{O(1)}$ as follows: 
Choose $w(j) \in \{0,1\}^n$ uniformly at random and compute $A_{F_n}(x,w(j))$ for every 
$1 \leq j \leq K$, and output $(1/K)\sum_{j=1}^K A_{F_n}(x,w(j))$.

The probabilistic algorithm $A_{F_n}$ in~\cite{Ni} can be considered as a repetition of a 
commuting quantum circuit $D_{2n}$ for the Hadamard test with $2n$ input qubits and 
one output qubit. For any $x,w \in\{0,1\}^n$, the output of $D_{2n}$ with 
the input qubits 
initialized to $|x\rangle|w\rangle$ and output qubit initialized to $|0\rangle$ is 0 
with probability $(1+F_n(x,w))/2$. Thus, when the 
outputs 0 and 1 are regarded as 1 and $-1$, respectively, due to the 
Chernoff-Hoeffding bound, $A_{F_n}$ is described with some $L=n^{O(1)}$ as 
follows, where the input is the pair of $x$ and $w$: Perform $D_{2n}$ with the input qubits 
initialized to $|x\rangle|w\rangle$ and output qubit initialized to $|0\rangle$, and obtain its 
output $z_j(x,w) \in \{1,-1\}$ for every $1 \leq j \leq L$. 
After that, output $(1/L)\sum_{j=1}^L z_j(x,w)$.

Our idea for proving Theorem~\ref{oracle} is to construct a parallelized version of the 
Hadamard test, denoted by $E_{2n}$, by using uninitialized ancillary qubits and replace 
$D_{2n}$ in the above algorithm for MAT$(p(n),C_n)$ with $E_{2n}$. Although 
the standard Hadamard test is a sequential application of controlled gates with the same 
control qubit, roughly speaking, $E_{2n}$ first prepares the copies of the state of 
the control qubit on uninitialized ancillary qubits and then applies the gates in parallel 
by using the copies. To be precise, let 
$x=x_1\cdots x_n, w=w_1\cdots w_n \in\{0,1\}^{n}$. We prepare $2n$ input 
qubits $X_1,\ldots,X_n, W_1,\ldots, W_n$, one output qubit $Y$, and $n$ uninitialized 
ancillary qubits $G(1),\ldots,G(n)$, where $X_j$, $W_j$, and $Y$ are initialized to 
$|x_j\rangle$, $|w_j\rangle$, and $|0\rangle$, respectively. The initial state of the 
uninitialized ancillary qubits is arbitrary. As an example, $E_{2n}$ with $n=3$ is given 
in Appendix~\ref{a-correct2}. The circuit $E_{2n}$ is defined as follows:
\begin{enumerate}
\item Apply a $H$ gate to $Y$.

\item Apply a fan-out gate on $n+1$ qubits to $G(1),\ldots,G(n)$, and $Y$, 
where $Y$ is the control~qubit.

\item Apply $C_n$ to $X_1,\ldots,X_{n-1}$, and $X_n$.

\item Apply a 2-controlled $Z$ gate to $G(j)$, $X_j$, and $W_j$ for every $1 \leq j \leq n$ 
in parallel.

\item Apply $C_n^\dag$ to $X_1,\ldots,X_{n-1}$, and $X_n$.

\item Apply the gates in Step~2 and Step~1 (in this order).
\end{enumerate}
Each fan-out gate can be decomposed into an $O(\log n)$-depth quantum 
circuit as described in Appendix~\ref{a-complexity}. 
Moreover, a 2-controlled $Z$ gate can be decomposed into a constant number of the gates 
in $\cal G$~\cite{Barenco,Nielsen}. Thus, $E_{2n}$ is an $O(\log n)$-depth circuit consisting 
of the gates in $\cal G$. It has the desired output probability distribution. 
The proof can be found in Appendix~\ref{a-correct2}.

\begin{lemma}\label{correct2}
For any $x,w \in\{0,1\}^{n}$, the output of $E_{2n}$ with the input qubits 
initialized to $|x\rangle|w\rangle$ and output qubit initialized to $|0\rangle$ is 0 with 
probability $(1+F_n(x,w))/2$.
\end{lemma}

This lemma immediately implies Theorem~\ref{oracle}:

\begin{proof}[Proof of Theorem~\ref{oracle}]
We replace $D_{2n}$ in the above-mentioned algorithm for MAT$(p(n),C_n)$ with 
$E_{2n}$. By Lemma~\ref{correct2}, the output probability distribution of 
$E_{2n}$ is the same as that of $D_{2n}$. Thus, as with the original algorithm, the resulting 
algorithm solves MAT$(p(n),C_n)$.
\end{proof}

\section{Limitations of Uninitialized Ancillary Qubits}\label{impossibility}

\subsection{Our Idea for Proving Theorem~\ref{cannot}}\label{idea2}

For any integer $s \geq 1$, an $s$-controlled Toffoli gate is decomposed into an 
$s$-controlled $Z$ gate sandwiched between two $H$ gates~\cite{Fang}. 
Thus, to prove Theorem~\ref{cannot}, it suffices to consider an unbounded $Z$ gate 
in place of an unbounded Toffoli gate. We assume on the contrary that there exists a 
depth-$d$ quantum circuit $C_n$ for PA$_n$ with $n$ input qubits, one output qubit, 
$p=O(\log n)$ initialized ancillary qubits, and $q=n^{O(1)}$ uninitialized ancillary qubits 
such that it consists of the gates in $\cal G$, unbounded fan-out gates on 
$(\log n)^{O(1)}$ qubits, and unbounded $Z$ gates, where $d=O((\log n)^{\delta})$ for 
some constant $0 \leq \delta <1$. When all unbounded $Z$ gates in $C_n$ act on a small 
number of qubits, such as $O(\log n)$ qubits, since $d$ is sufficiently small, the proof 
of Bera~\cite{Bera2} implies that there exists an input qubit of $C_n$ such 
that the output of $C_n$ does not depend on the input 
qubit. Thus, $C_n$ cannot compute PA$_n$ since the output of PA$_n$ 
changes if any one of the $n$ input bits changes. This contradicts the assumption.

The remaining case is when there exists an unbounded $Z$ gate on a large number of qubits. 
Let $\tilde C_n$ be the circuit obtained from $C_n$ by removing all such gates. 
Bera~\cite{Bera2} showed that, when $C_n$ does not have any ancillary qubit, 
it is well approximated by $\tilde C_n$ in the sense that, when the state of the 
input qubits is a computational basis state chosen uniformly at random, the output of 
$C_n$ coincides with that of $\tilde C_n$ with high probability. Since $C_n$ computes PA$_n$, 
$\tilde C_n$ computes PA$_n$ with high probability. Thus, we obtain a contradiction as 
in the above case since all gates in $\tilde C_n$ act on a small 
number of qubits. To apply this idea to our setting, we show that $C_n$ with $p$ initialized ancillary 
qubits and $q$ uninitialized ancillary qubits in state $|a\rangle$ for some $a \in \{0,1\}^q$ is 
well approximated (in the sense described above) by $\tilde C_n$ with the same state. 
The former circuit computes PA$_n$ since $C_n$ with an arbitrary initial state 
of the uninitialized ancillary qubits computes PA$_n$. Thus, the latter circuit computes 
PA$_n$ with high probability, and we obtain a contradiction as in the above simple case.

\subsection{Analysis of a General Circuit and Its Application}\label{analysis}

We analyze a general depth-$d$ quantum circuit $C_n$ with $n$ input qubits, one output 
qubit, and $p$ initialized and $q$ uninitialized ancillary qubits such that it 
consists of the gates in $\cal G$, unbounded fan-out gates, and unbounded $Z$ gates. 
Its key property is described as~follows:

\begin{figure}[t]
\centering
\includegraphics[scale=.32]{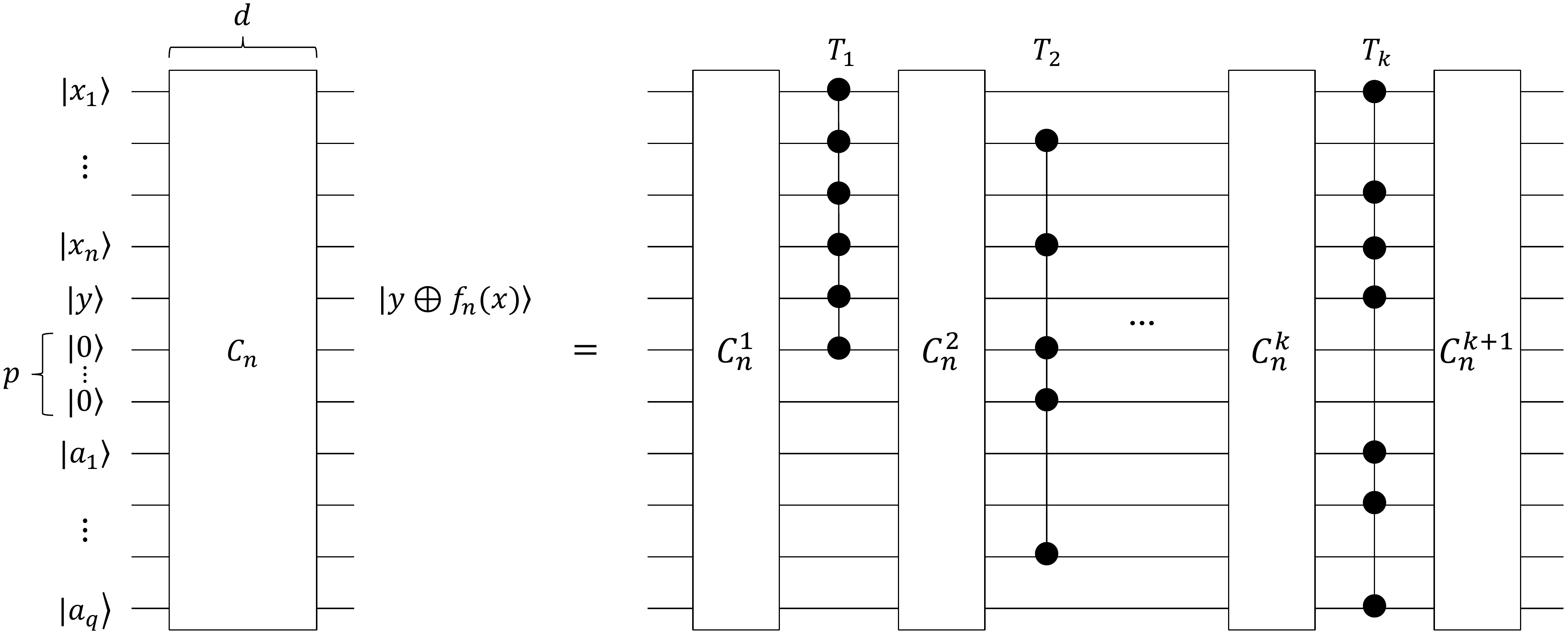}
\caption{Circuit $C_n$ for $f_n$ and its decomposition. The initial states of the input 
qubits, output qubit, and uninitialized ancillary qubits are 
$|x_1\rangle\cdots |x_n\rangle$, $|y\rangle$, and $|a_1\rangle\cdots |a_q\rangle$, respectively, 
for any $x=x_1\cdots x_n \in \{0,1\}^n$, $y \in \{0,1\}$, and $a_1\cdots a_q \in \{0,1\}^q$. 
Gates $T_1,\ldots,T_k$ are unbounded $Z$ gates.}
\label{figure5}
\end{figure}

\begin{lemma}[\cite{Bera2,Beraphd}]\label{Beralemma}
Let $C_n$ be a depth-$d$ quantum circuit with $n$ input qubits and one output qubit 
(possibly with ancillary qubits). If all gates in $C_n$ act on at most $w$ qubits, 
then the output of $C_n$ can depend only on the states of at most $w^d$ input qubits.
\end{lemma}

Let $t \geq 2$ be an integer and ${\cal G}_t$ be the set of all unbounded $Z$ gates 
in $C_n$ that act on more than or equal to $t$ qubits. We consider the case where 
${\cal G}_t \neq \emptyset$ and assume that ${\cal G}_t =\{T_1,\ldots,T_k\}$ for some $k \geq 1$, 
where, for any $1 \leq l \leq k$, if $T_l$ is in layer $L$ of $C_n$, then $T_{l+1}$ is in 
layer $L' \geq L$. We decompose $C_n$ into the gates in ${\cal G}_t$ and the other parts 
as depicted in Fig.~\ref{figure5}, where $C_n$ computes a Boolean function $f_n$ on $n$ 
bits and $C_n^j$ is a quantum circuit consisting of gates that 
are not in ${\cal G}_t$ for any $1 \leq j \leq k+1$. Such a decomposition is not 
unique in general, but the point is to fix a decomposition. 
For any $1 \leq l \leq k$,  we define a quantum circuit $V_l$ as follows: 
$V_1=C_n^1$ and $V_l = C_n^lT_{l-1}V_{l-1}$ for any $2 \leq l \leq k$. We also define 
$\Delta_l(x,y,b) = ||T_lV_l|x\circ y \circ b\rangle - V_l|x\circ y \circ b\rangle||$ 
and $\Delta(x,y,b) = ||C_n|x\circ y \circ b\rangle - {\tilde C_n}|x\circ y \circ b\rangle||$ 
for any $x\in\{0,1\}^n$, $y\in\{0,1\}$, and $b\in\{0,1\}^{p+q}$. Here, 
the symbol ``$\circ$'' represents the concatenation of bit strings, 
$|||v\rangle||= \sqrt{\langle v|v\rangle}$ for any vector 
$|v\rangle$, and ${\tilde C_n} =C_n^{k+1}C_n^k \cdots C_n^2C_n^1$. 
Let $U_n$ be a random variable uniformly distributed over $\{0,1\}^n$.

Using the expected value ${\rm E}[\Delta_l(U_n,y,b)^2] = (1/2^n)\sum_{x\in\{0,1\}^n} \Delta_l(x,y,b)^2$, we first evaluate the probability ${\rm Pr}[\Delta(U_n,y,b) < \varepsilon]$ 
as follows. The proof can be found in Appendix~\ref{a-union}.

\begin{lemma}\label{union}
${\rm Pr}[\Delta(U_n,y,b) < \varepsilon] \geq 
1-(k^2/\varepsilon^2)\sum_{l=1}^k{\rm E}[\Delta_l(U_n,y,b)^2]$ for any 
$\varepsilon >0$, $y\in\{0,1\}$, and $b\in\{0,1\}^{p+q}$.
\end{lemma}

To evaluate the value $\sum_{l=1}^k{\rm E}[\Delta_l(U_n,y,b)^2]$, let $t_l$ be the number 
of qubits on which $T_l$ acts, $u_l=n+p+q+1-t_l$, and 
$t_{\rm min} = \min\{t_l|1 \leq l \leq k\}$. We assume 
that $V_l|x\circ y\circ b\rangle = 
\sum_{i\in\{0,1\}^{t_l}}\sum_{j\in\{0,1\}^{u_l}} g_{x\circ y\circ b}^{(l)}(i \circ j)|i\circ j\rangle$ 
for any $x\in\{0,1\}^n$, $y\in\{0,1\}$, and $b\in\{0,1\}^{p+q}$, 
where $g_{x\circ y\circ b}^{(l)}(i \circ j)$ is a complex number. The qubits represented by 
$i\in\{0,1\}^{t_l}$ correspond to the qubits on which $T_l$ acts. Of course, 
for any $1 \leq l \leq k$, $T_l$ does not always act on the first $t_l$ qubits in $C_n$. We 
therefore need to apply some permutation of all qubits; however, since such a permutation does 
not affect Lemma~\ref{union2}, which is the key to Theorem~\ref{cannot}, we omit it.
 
We evaluate the above value as follows. The point is that this value with some 
initial state of the uninitialized ancillary qubits is small. The proof can be found 
in Appendix~\ref{a-bunsu}.

\begin{lemma}\label{bunsu}
$\sum_{l=1}^k {\rm E}[\Delta_l(U_n,y,b)^2]
\leq k2^{p+q+3}/2^{t_{\rm min}}$ for any $y\in\{0,1\}$ and $b\in\{0,1\}^{p+q}$. Moreover, 
there exists some $a \in \{0,1\}^q$ such that $\sum_{l=1}^k 
{\rm E}[\Delta_l(U_n,0,0^p\circ a)^2] \leq k2^{p+3}/2^{t_{\rm min}}$.
\end{lemma}

Lemmas~\ref{union} and~\ref{bunsu} immediately imply the following evaluation:

\begin{lemma}\label{union2}
There exists some $a \in \{0,1\}^q$ such that ${\rm Pr}[\Delta(U_n,0,0^p\circ a) < 
\varepsilon] \geq 1- k^3 2^{p+3}/(\varepsilon^2{2^{t_{\rm min}}})$ for any $\varepsilon >0$. 
\end{lemma}

Lemmas~\ref{Beralemma} and~\ref{union2} imply Theorem~\ref{cannot} as follows:

\begin{proof}[Proof of Theorem~\ref{cannot}]

We assume on the contrary that there exists a quantum circuit $C_n$ for PA$_n$ 
described in Section~\ref{idea2}. Since $p=O(\log n)$, there exists a constant $c > 0 $ 
such that $p \leq c \log n$ when $n$ is sufficiently large. We define $t=(c+4)\log (n+p+q+1)$ 
and consider ${\cal G}_t$ described above. When ${\cal G}_t = \emptyset$, all gates in 
$C_n$ act on at most $w = (\log n)^{O(1)}$ qubits. By Lemma~\ref{Beralemma}, the output 
of $C_n$ can depend only on the states of at most $w^d=o(n)$ input qubits. Thus, 
there exists an input qubit of $C_n$ such that the output of $C_n$ does not depend on 
the input qubit. This yields a contradiction as described in Section~\ref{idea2}. 

We consider the remaining case where ${\cal G}_t \neq \emptyset$. In this case, 
we apply the above analysis of a general circuit. 
It holds that $p\leq c\log n$, $k\leq (n+p+q+1)d/t_{\rm min}$, and 
$t_{\rm min} \geq (c+4)\log (n+p+q+1)$. Thus, by Lemma~\ref{union2} with $\varepsilon=0.1$, 
\begin{align*}
{\rm Pr}[\Delta(U_n,0,0^{p}\circ a) < 0.1] 
\geq  1-\left( \frac{d}{(c+4)\log (n+p+q+1)} \right)^3\frac{800n^c}{(n+p+q+1)^{c+1}}
\end{align*}
for some $a \in \{0,1\}^q$. Let us express this value on the right-hand side 
by $1-\gamma$. Thus, there exists a set 
$S\subseteq \{0,1\}^n$ such that $S$ has at least $2^n(1-\gamma)$ 
elements and, for any $x \in S$, $\Delta(x,0,0^{p}\circ a) < 0.1$. Since 
$\gamma$ goes to 0 as $n$ goes to infinity, 
$2^n(1-\gamma) > 2^{n-1}$ when $n$ is sufficiently large. A simple calculation shows that, 
for any $x \in \{0,1\}^n$ satisfying $\Delta(x,0,0^{p}\circ a) < 0.1$, the 
output of ${\tilde C_n}|x\circ 0 \circ 0^{p}\circ a\rangle$ coincides with 
that of $C_n|x\circ 0 \circ 0^{p}\circ a\rangle$ with probability of 
at least $1-0.1^2 = 0.99$~\cite{Bera2,Beraphd}. When the initial state of the 
uninitialized ancillary qubits is $|a\rangle$, $C_n$ computes PA$_n$. Thus, for any $x \in S$, 
the output of ${\tilde C_n}|x\circ 0 \circ 0^{p}\circ a\rangle$ is PA$_n(x)$ with probability of 
at least $0.99$. This contradicts the fact obtained by the following argument. 
Since all gates in $\tilde C_n$ act on at 
most $(\log n)^{O(1)}$ qubits, as described for the case 
where ${\cal G}_t =\emptyset$, by Lemma~\ref{Beralemma}, there exists an input qubit of 
$\tilde C_n$ such that the output of $\tilde C_n$ does not 
depend on the input qubit. This implies that, for at most $2^{n-1}$ 
elements $x\in\{0,1\}^n$, the output of ${\tilde C_n}|x\circ 0 \circ 0^{p}\circ a\rangle$ 
is PA$_n(x)$ with probability greater than 0.5.
\end{proof}

\section{Open Problems}

Interesting challenges would be to further study the computational power of shallow
quantum circuits with uninitialized ancillary qubits. We give some examples of such problems:
\begin{itemize}
\item Can we decrease the depth of the circuit in Theorem~\ref{can}?

\item What (non-symmetric) functions can be computed by shallow 
quantum circuits with uninitialized ancillary qubits?

\item What is the relationship between the computational power of 
shallow quantum circuits with uninitialized ancillary qubits and that of 
general classical/quantum circuits?
\end{itemize}

\bibliography{mybib}

\appendix

\section{Proofs}

\subsection{Proof of Lemma~\ref{correct}}\label{a-correct}

\begin{proof}
As an example, Stages 1 and 2 with $n=3$ are depicted in 
Figs.~\ref{figure2} and~\ref{figure3}, respectively. The states of $X_1,\ldots,X_n$ stay unchanged 
during the computation since the qubits are used only for control qubits. 
We fix an arbitrary $1 \leq k \leq m$ and show the lemma by induction on $s$. 
We first consider the base case, $s=1$. Steps 1--2 transform the initial state of 
$I(k)$ and $B_1(k,1),\ldots,B_n(k,1)$, which is $|0\rangle |b_1(k,1)\rangle\cdots |b_n(k,1)\rangle$, 
into the state 
\begin{equation}\label{sta0}
\frac{1}{\sqrt{2}} |0\rangle |b_1(k,1)\rangle\cdots |b_n(k,1)\rangle
+ \frac{1}{\sqrt{2}} |1\rangle |b_1(k,1)\oplus 1 \rangle\cdots |b_n(k,1)\oplus 1\rangle.
\end{equation}
Step 3 does nothing and 
Step 4 transforms the states of $A_1(k),\ldots,A_n(k)$ into the states 
$|x_1\oplus a_1(k)\rangle,\ldots,|x_n\oplus a_n(k)\rangle$, respectively. 
Step 5 transforms state (\ref{sta0}) into the state
\begin{equation}\label{sta1}
\frac{1}{\sqrt{2}} |0\rangle |b_1(k,1)\rangle\cdots |b_n(k,1)\rangle + 
\frac{e^{\frac{2\pi i}{2^k}\alpha(k,1)}}{\sqrt{2}}|1\rangle 
|b_1(k,1)\oplus 1 \rangle\cdots |b_n(k,1)\oplus 1\rangle,
\end{equation}
where we ignore the global phase and
\begin{align*} 
\alpha(k,1)=\sum_{j=1}^n (-1)^{b_j(k,1)}(x_j \oplus a_j(k)).
\end{align*} 
Step 6 transforms the states of $A_1(k),\ldots,A_n(k)$ into the states 
$|a_1(k)\rangle,\ldots,|a_n(k)\rangle$, respectively. Step~7 transforms 
state (\ref{sta1}) into the state
\begin{align*}
\frac{1}{\sqrt{2}} |0\rangle |b_1(k,1)\rangle\cdots |b_n(k,1)\rangle + 
\frac{e^{\frac{2\pi i}{2^k}(\alpha(k,1)+\beta(k,1))}}{\sqrt{2}}|1\rangle 
|b_1(k,1)\oplus 1 \rangle\cdots |b_n(k,1)\oplus 1\rangle,
\end{align*}
where
\begin{align*}
\beta(k,1)=-\sum_{j=1}^n(-1)^{b_j(k,1)}a_j(k).
\end{align*} 
Thus, $\alpha(k,1) + \beta(k,1)$ is equal to the following value:
\begin{align*}
& \sum_{j=1}^n(-1)^{b_j(k,1)}((x_j \oplus a_j(k))-a_j(k))
=\sum_{j=1}^n (-1)^{b_j(k,1)}(-1)^{a_j(k)}x_j\\
& =\sum_{j=1}^n x_j(1 - 2 (a_j(k) \oplus b_j(k,1)))
=|x| - 2\sum_{j=1}^nx_j(a_j(k) \oplus b_j(k,1)) = \gamma(k,1).
\end{align*}
Steps 8 transforms the state of all the qubits other than $I(k)$ into their initial state. 
The state of $I(k)$ is
\begin{equation}\label{after}
\frac{|+\rangle + e^{\frac{2\pi i}{2^k}\gamma(k,1)}|-\rangle}{\sqrt{2}}.
\end{equation}
The state $|\varphi_2'\rangle$ in Figs.~\ref{figure2} and~\ref{figure3} is state (\ref{after}) 
when $n=3$ and $k=2$. This completes the proof of the base case. In particular, the above 
proof implies that Lemma~\ref{correct} holds when $k=1$, and thus 
we assume that $k \geq 2$ in the following.

\begin{figure}
\centering
\includegraphics[scale=.32]{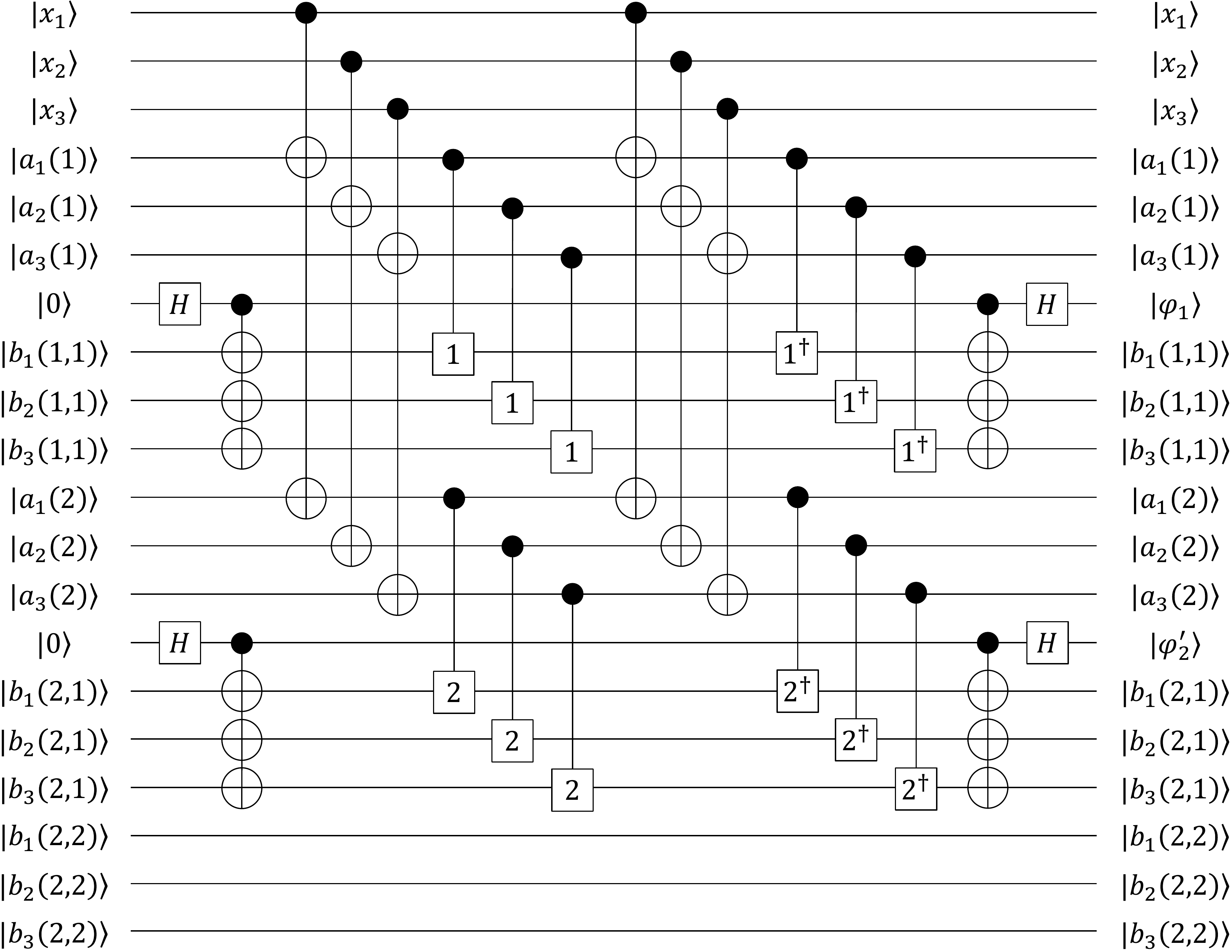}
\caption{Stage 1.}
\label{figure2}
\end{figure}
\begin{figure}
\centering
\includegraphics[scale=.32]{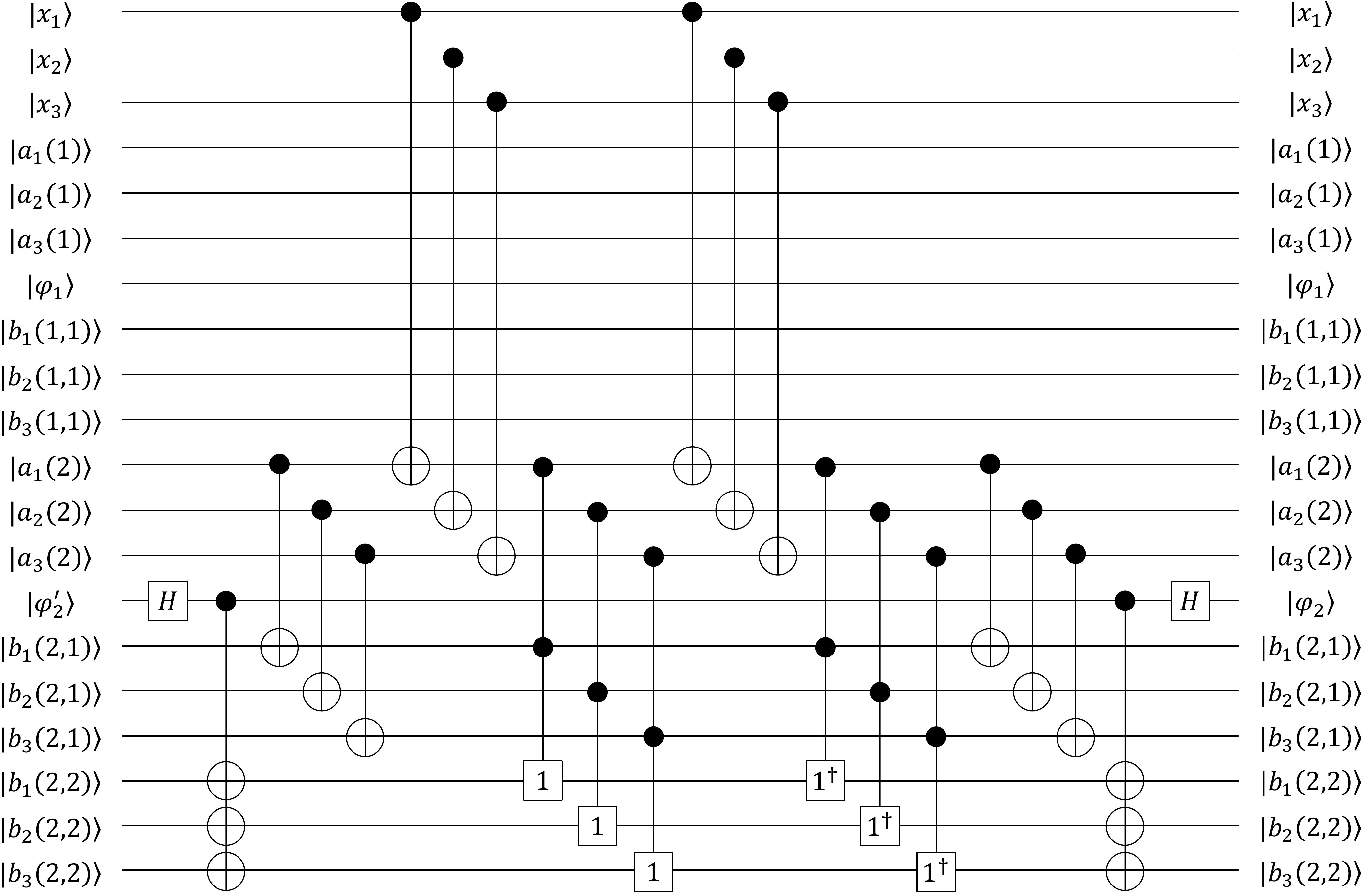}
\caption{Stage 2.}
\label{figure3}
\end{figure}

We fix an arbitrary $1 \leq s \leq k-1$ and assume that the lemma holds for $s$. 
Thus, after Stage $s$, the state of $I(k)$ is
\begin{align*}
\frac{|+\rangle + e^{\frac{2\pi i}{2^k}\gamma(k,s)}|-\rangle}{\sqrt{2}}.
\end{align*}
Moreover, for any qubit other than the initialized ancillary qubits, the state of the qubit is the 
same as its initial state. We apply Stage $s+1 \leq k$. The current state of $I(k)$ and 
$B_1(k,s+1),\ldots,B_n(k,s+1)$ is
\begin{align*}
\frac{|+\rangle + e^{\frac{2\pi i}{2^k}\gamma(k,s)}|-\rangle}{\sqrt{2}}
|b_1(k,s+1)\rangle\cdots |b_n(k,s+1)\rangle
\end{align*}
and Steps 1--2 transform this state into the state
\begin{equation}\label{sta2}
\frac{1}{\sqrt{2}} |0\rangle |b_1(k,s+1)\rangle\cdots |b_n(k,s+1)\rangle + 
\frac{e^{\frac{2\pi i}{2^k}\gamma(k,s)}}{\sqrt{2}} |1\rangle 
|b_1(k,s+1)\oplus 1 \rangle\cdots |b_n(k,s+1)\oplus 1\rangle.
\end{equation}
Step 3 transforms the states of $B_j(k,1),\ldots,B_j(k,s)$ 
into the states
\begin{align*}
|a_j(k)\oplus b_j(k,1)\rangle,\ldots,|a_j(k)\oplus b_j(k,s)\rangle,
\end{align*} 
respectively, for any $1 \leq j \leq n$. Step 4 transforms the states of 
$A_1(k),\ldots,A_n(k)$ into the states $|x_1 \oplus a_1(k)\rangle,\ldots,|x_n \oplus a_n(k)\rangle$, 
respectively. Step 5 transforms state (\ref{sta2}) into the state
\begin{equation}\label{sta3}
\frac{1}{\sqrt{2}} |0\rangle |b_1(k,s+1)\rangle\cdots |b_n(k,s+1)\rangle + 
\frac{e^{\frac{2\pi i}{2^k}(\gamma(k,s)+\alpha(k,s+1))}}{\sqrt{2}} |1\rangle 
|b_1(k,s+1)\oplus 1 \rangle\cdots |b_n(k,s+1)\oplus 1\rangle,
\end{equation}
where
\begin{align*}
\alpha(k,s+1) 
=2^s\sum_{j=1}^n (-1)^{b_j(k,s+1)}(x_j\oplus a_j(k))\bigwedge_{l=1}^{s}  (a_j(k) \oplus b_j(k,l)).
\end{align*} 
Step 6 transforms the states of 
$A_1(k),\ldots,A_n(k)$ into the states $|a_1(k)\rangle,\ldots,|a_n(k)\rangle$, 
respectively. Step~7 transforms state (\ref{sta3}) into the state
\begin{align*}
\frac{1}{\sqrt{2}} |0\rangle & |b_1(k,s+1)\rangle \cdots |b_n(k,s+1)\rangle\\
& + \frac{e^{\frac{2\pi i}{2^k}(\gamma(k,s)+\alpha(k,s+1)+\beta(k,s+1))}}{\sqrt{2}} |1\rangle 
|b_1(k,s+1)\oplus 1 \rangle\cdots |b_n(k,s+1)\oplus 1\rangle,
\end{align*}
where
\begin{align*}
\beta(k,s+1) =-2^s\sum_{j=1}^n(-1)^{b_j(k,s+1)}a_j(k)\bigwedge_{l=1}^{s}  (a_j(k) \oplus b_j(k,l)).
\end{align*} 
Thus, $\alpha(k,s+1) + \beta(k,s+1)$ is equal to the following value:
\begin{align*}
& 2^s\sum_{j=1}^n (-1)^{b_j(k,s+1)}((x_j \oplus a_j(k))-a_j(k))
\bigwedge_{l=1}^{s}  (a_j(k) \oplus b_j(k,l))\\
&=2^s\sum_{j=1}^n (-1)^{b_j(k,s+1)}(-1)^{a_j(k)}x_j\bigwedge_{l=1}^{s}  (a_j(k) \oplus b_j(k,l))\\
&=2^s\sum_{j=1}^n x_j(1 - 2 (a_j(k) \oplus b_j(k,s+1)))\bigwedge_{l=1}^{s}  (a_j(k) \oplus b_j(k,l))\\
&=2^s\sum_{j=1}^n x_j\bigwedge_{l=1}^{s}  (a_j(k) \oplus b_j(k,l))
- 2^{s+1}\sum_{j=1}^nx_j\bigwedge_{l=1}^{s+1}  (a_j(k) \oplus b_j(k,l))\\
&=|x|-\gamma(k,s)- 2^{s+1}\sum_{j=1}^nx_j\bigwedge_{l=1}^{s+1}  (a_j(k) \oplus b_j(k,l))
=\gamma(k,s+1)-\gamma(k,s).
\end{align*}
Thus, $\gamma(k,s)+\alpha(k,s+1) + \beta(k,s+1)= \gamma(k,s+1)$. 
Step 8 transforms the state of all the qubits other than $I(k)$ into their initial state. 
The state of $I(k)$ is
\begin{align*}
\frac{|+\rangle + e^{\frac{2\pi i}{2^k}\gamma(k,s+1)}|-\rangle}{\sqrt{2}}.
\end{align*}
Therefore, the lemma holds for $s+1$ as desired. We note that
\begin{align*}
e^{\frac{2\pi i}{2^k}\gamma(k,k)} = e^{\frac{2\pi i}{2^k}(|x| - 2^k \sum_{j=1}^n
x_j \bigwedge_{l=1}^{k} (a_j(k)\oplus b_j(k,l)))} = e^{\frac{2\pi i}{2^k}|x|}.
\end{align*}
This completes the proof of the lemma.
\end{proof}

\subsection{Proof of Lemma~\ref{complexity}}\label{a-complexity}

\begin{proof}
We first decompose a $k$-controlled $Z(\pm 2\pi/2^t)$ gate in Steps 5 and 7 into the gates 
in $\cal G$ for any integers $k \geq 2$ and $t \geq 1$. Decomposing each gate simply by the 
standard method~\cite{Barenco} does not yield a quantum circuit with 
the desired complexity and thus we use a structure of these steps. These steps can be 
considered as a parallel application of fan-out gates (Step 6) preceded and followed by controlled 
phase-shift gates (Steps 5 and 7). We focus on a part of the structure, which can be 
represented as a fan-out gate preceded by a $k$-controlled $Z(2\pi/2^t)$ gate and followed by 
a $k$-controlled $Z(-2\pi/2^t)$ gate. As an example, when $k=8$, the circuit 
is the leftmost one depicted in Fig.~\ref{figure4}. By 
the standard decomposition method, a $k$-controlled $Z(2\pi/2^t)$ gate is 
decomposed into a $(k-1)$-controlled $Z(2\pi/2^{t+1})$ gate, two $(k-1)$-controlled Toffoli 
gates, a 1-controlled $Z(2\pi/2^{t+1})$ gate, and a 1-controlled $Z(-2\pi/2^{t+1})$ gate. 
Similarly, a $k$-controlled $Z(-2\pi/2^t)$ gate is decomposed into a 
$(k-1)$-controlled $Z(-2\pi/2^{t+1})$ gate, two $(k-1)$-controlled Toffoli 
gates, a 1-controlled $Z(2\pi/2^{t+1})$ gate, and a 1-controlled $Z(-2\pi/2^{t+1})$ gate. 
The resulting circuit with $k=8$ is the middle one depicted in Fig.~\ref{figure4}. In 
the resulting circuit, the $(k-1)$-controlled $Z(2\pi/2^{t+1})$ gate is canceled out by 
the $(k-1)$-controlled $Z(-2\pi/2^{t+1})$ gate. Similarly, a $(k-1)$-controlled Toffoli 
gate is canceled out by another one. Thus, the remaining gates 
(other than the fan-out gate) are two $(k-1)$-controlled Toffoli gates, 
two 1-controlled $Z(2\pi/2^{t+1})$ gates, and two 1-controlled $Z(-2\pi/2^{t+1})$ gates. 
The final circuit with $k=8$ is the rightmost one depicted in Fig.~\ref{figure4}. For any 
integer $r \geq 1$, a $1$-controlled $Z(2\pi/2^r)$ gate is decomposed into the gates in 
$\cal G$ as depicted in Fig.~\ref{figure0}(a). A $1$-controlled $Z(-2\pi/2^r)$ gate is 
decomposed similarly. Thus, the remaining problem is to decompose 
an unbounded Toffoli gate and an unbounded fan-out gate.

A $k$-controlled Toffoli gate is decomposed into an $O(k)$-depth quantum 
circuit with an uninitialized ancillary qubit such that it consists of $H$ gates, 
$Z(\pm\pi/4)$ gates, and CNOT gates \cite{Barenco,Nielsen}. Moreover, on the basis of 
the fact that a fan-out gate on $k+1$ qubits is equivalent to a gate for computing PA$_k$ 
sandwiched between two layers of $H$ gates~\cite{Green}, it 
is decomposed into an $O(\log k)$-depth quantum circuit without using any new ancillary 
qubit such that it consists of CNOT gates. An example of such a circuit with $k=4$ is depicted in 
Fig.~\ref{figure0}(b). By these decompositions, we can regard $Q_n$ as a 
circuit consisting of the gates in $\cal G$. The depth of each stage is 
$O(\log n)$ since an unbounded fan-out gate acts on at most $n+1=O(n)$ qubits and 
an unbounded Toffoli gate acts on at most $m+1=O(\log n)$ qubits. Thus, the depth of 
the whole circuit is $O(m\log n)=O((\log n)^2)$. Moreover, 
it uses $m=O(\log n)$ initialized ancillary qubits and $nm(m+3)/2=O(n(\log n)^2)$ 
uninitialized ancillary qubits. 
\end{proof}

\begin{figure}[t]
\centering
\includegraphics[scale=.32]{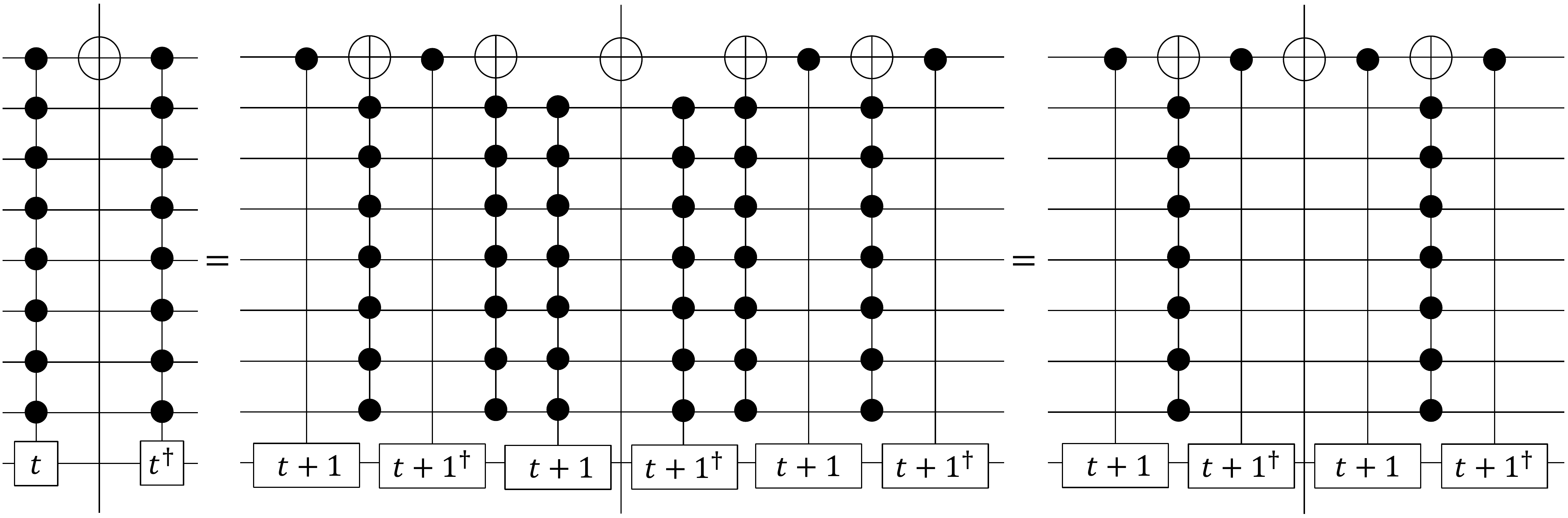}
\caption{Decomposition of the leftmost circuit consisting of a fan-out gate, an 
8-controlled $Z(2\pi/2^t)$ gate, and an 8-controlled $Z(-2\pi/2^t)$ gate. 
The gate between these controlled phase-shift gates represents a part of the 
fan-out gate. The middle circuit is obtained by the standard method~\cite{Barenco} for 
decomposing the controlled phase-shift gates. 
The rightmost circuit is obtained from the middle one since 
the 7-controlled $Z(2\pi/2^{t+1})$ gate is canceled out by the 7-controlled 
$Z(-2\pi/2^{t+1})$ gate and a 7-controlled Toffoli gate is canceled out by another one.}
\label{figure4}
\end{figure}
\begin{figure}[t]
\centering
\includegraphics[scale=.32]{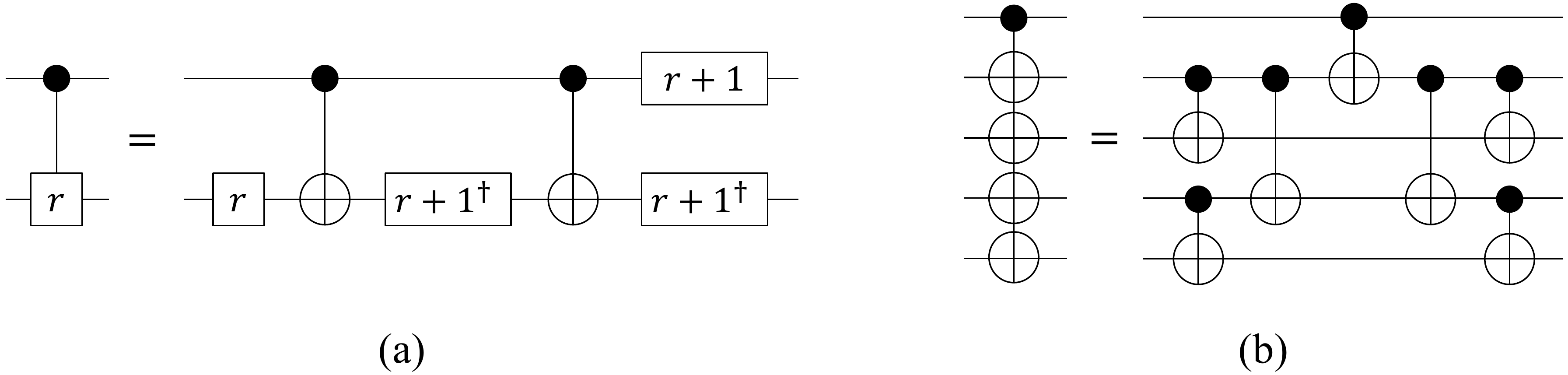}
\caption{(a): Decomposition of a 1-controlled $Z(2\pi/2^r)$ gate for any integer 
$r \geq 1$. (b): Decomposition of a fan-out gate on five qubits.}
\label{figure0}
\end{figure}

\subsection{Construction of $R_m$ and the Proof of Lemma~\ref{exp}}\label{a-exp}

We first describe the construction of $R_m$ with input $s=s_1\cdots s_m\in\{0,1\}^m$ under the 
assumption that we have a gate, which we call a PARITY$(m)$ gate, that implements the 
operation on $2^m+m-1$ qubits defined as
\begin{align*}
|s\rangle\bigotimes_{t\in\{0,1\}^m \setminus \{0^m\}}|a_t\rangle \mapsto 
|s\rangle\bigotimes _{t\in\{0,1\}^m \setminus \{0^m\}}|a_t \oplus \bigoplus_{k=1}^mt_ks_k\rangle
\end{align*}
for any $s=s_1\cdots s_m\in\{0,1\}^m$ and $a_t\in\{0,1\}$. This gate will be decomposed 
into the gates in $\cal G$ later. We prepare $m$ input qubits $S_1,\ldots,S_m$ and one 
output qubit $Y$, where $S_k$ is 
initialized to $|s_k\rangle$ and $Y$ is initialized to $|y\rangle$ for any $y \in \{0,1\}$. 
We also prepare $(m+1)(2^m-1)$ uninitialized ancillary qubits, 
which are divided into two groups, $A$ and $B$. Group $A$ consists of $2^m-1$ qubits, 
each of which is represented as $A_t$ for any $t\in\{0,1\}^m \setminus \{0^m\}$, where 
the initial state of $A_t$ is $|a_t\rangle$ for any (unknown) $a_t \in\{0,1\}$. Group $B$ 
consists of $m(2^m-1)$ qubits, which are divided into 
$m$ groups $B(1),\ldots,B(m)$. Each $B(l)$ consists of $2^m-1$ qubits, each of which 
is represented as $B_t(l)$ for any $t\in\{0,1\}^m \setminus \{0^m\}$, where the initial 
state of $B_t(l)$ is $|b_t(l)\rangle$ for any (unknown) $b_t(l)\in\{0,1\}$. 

\begin{figure}[t]
\centering
\includegraphics[scale=.32]{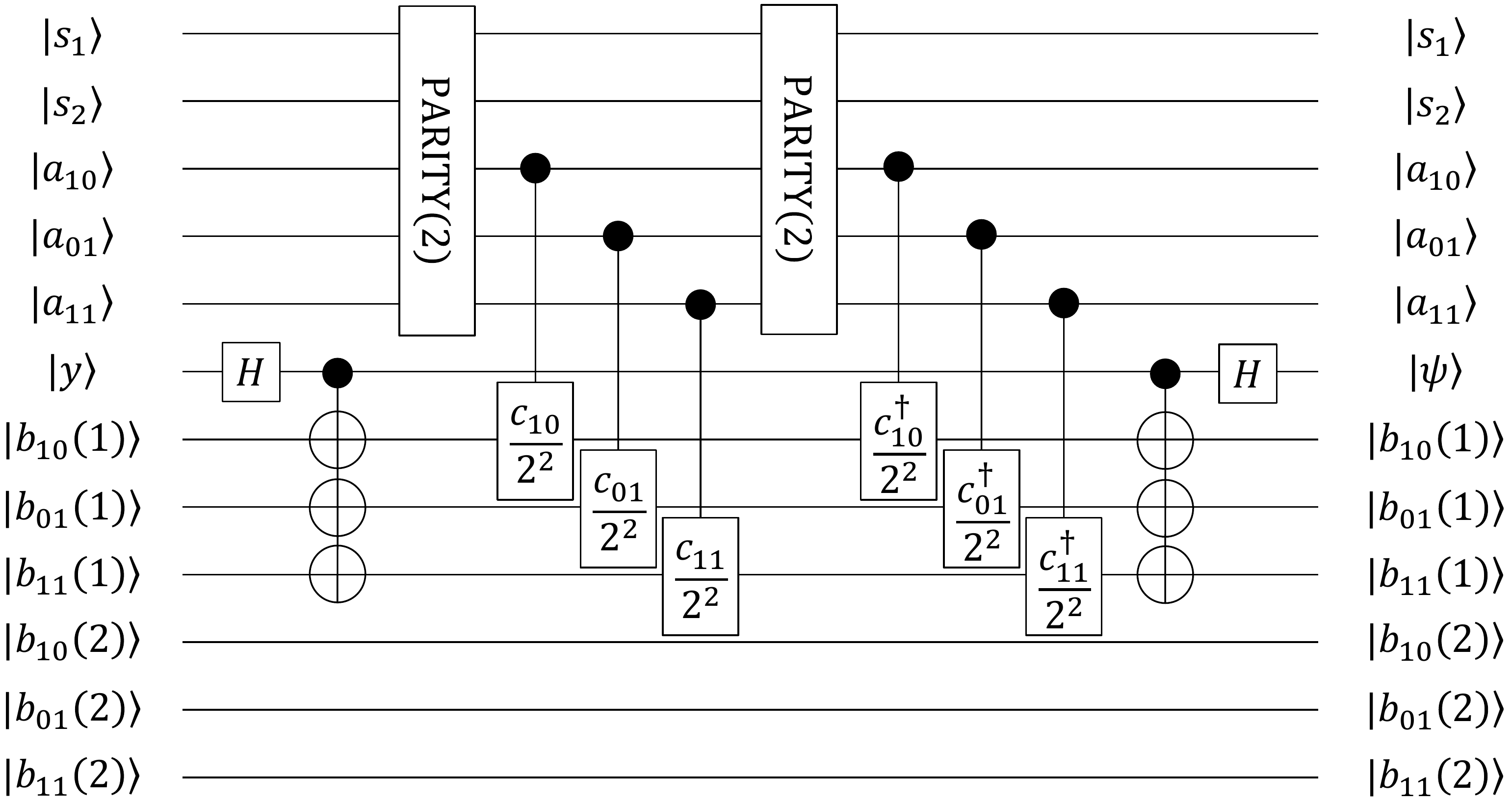}
\caption{Stage 1 of $R_2$. For any $t \in \{10,01,11\}$ and integer $l \geq 1$, 
the gates $c_t/2^l$ and $c_t^{\dagger}/2^l$ represent a 
$Z(2\pi c_t/2^l)$ gate and its inverse, i.e., a $Z(-2\pi c_t/2^l)$ gate, respectively.}
\label{figure7}
\end{figure}
\begin{figure}
\centering
\includegraphics[scale=.32]{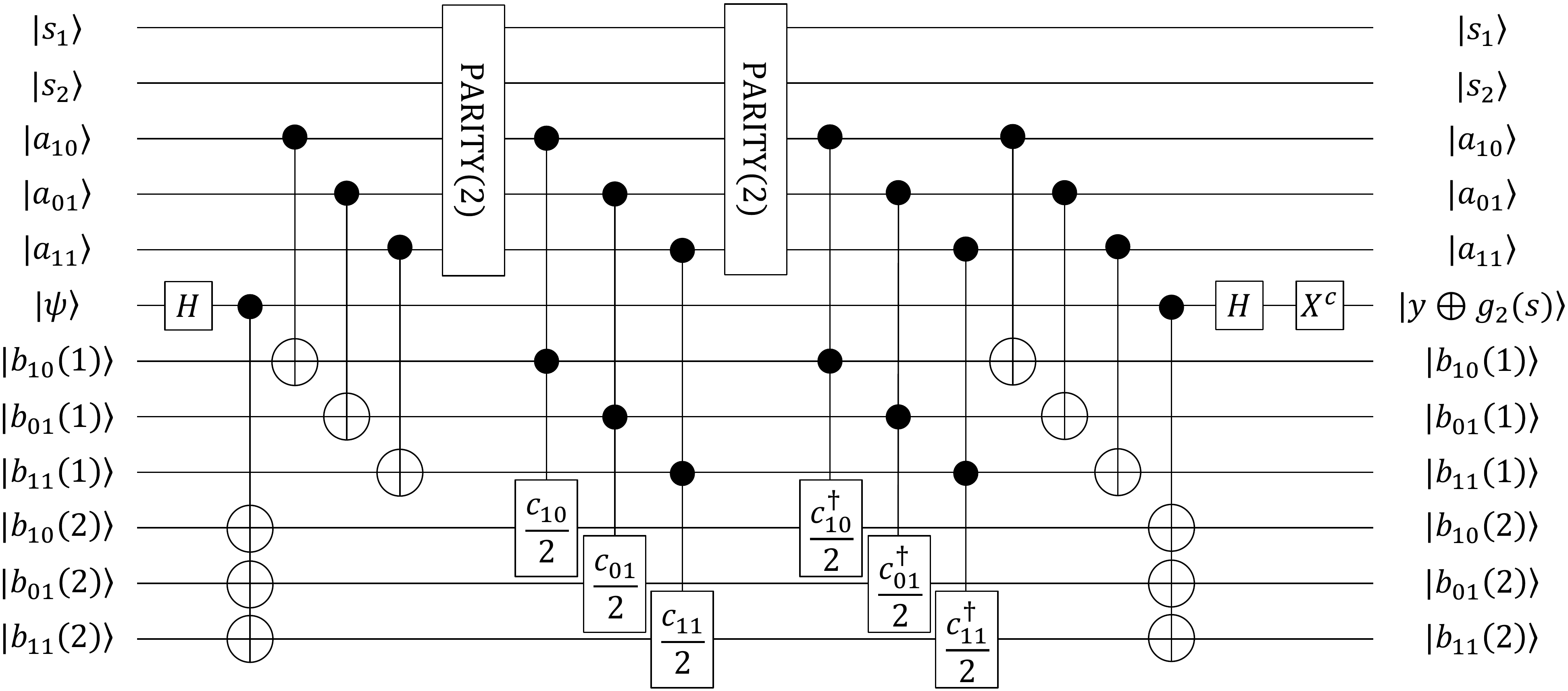}
\caption{Stage 2 of $R_2$ followed by the final $X^c$ gate, where $c=g_2(0^2)$.}
\label{figure8}
\end{figure}

The circuit $R_m$ consists of $m$ stages and a final gate after Stage $m$. As an example, 
Stages 1 and 2 followed by the final gate with $m=2$ are depicted in 
Figs.~\ref{figure7} and~\ref{figure8}, respectively. In these figures,
\begin{align*}
|\psi\rangle = X^y \left(\frac{|+\rangle + e^{\frac{2\pi i}{2^2}\gamma}|-\rangle}{\sqrt{2}}\right),
\end{align*}
where
\begin{align*}
\gamma = \sum_{t\in\{0,1\}^2 \setminus \{0^2\}}c_t\bigoplus_{k=1}^2t_ks_k
- 2\sum_{t\in\{0,1\}^2 \setminus \{0^2\}} c_t\bigoplus_{k=1}^2t_ks_k(a_t\oplus b_t(1)).
\end{align*} 
For any $1\leq u \leq m$, Stage $u$ is defined as follows:
\begin{enumerate}
\item Apply a $H$ gate to $Y$.

\item Apply a fan-out gate on $2^m$ qubits to $Y$ and all $B_t(u)$'s with 
$t\in\{0,1\}^m \setminus \{0^m\}$, where $Y$ is the control qubit.

\item If $u \geq 2$, then apply a fan-out gate on $u$ qubits to 
$B_t(1),\ldots,B_t(u-1)$, and $A_t$ for every $t\in\{0,1\}^m \setminus \{0^m\}$ in parallel, 
where $A_t$ is the control qubit.

\item Apply a PARITY$(m)$ gate on $2^m+m-1$ qubits to 
$S_1,\cdots,S_m$, and all $A_t$'s with $t\in\{0,1\}^m \setminus \{0^m\}$.

\item Apply a $u$-controlled $Z(2\pi c_t/2^{m-u+1})$ gate to $B_t(u)$ and the following 
qubits for every $t\in\{0,1\}^m \setminus \{0^m\}$ in parallel: 
$A_t$ if $u=1$ and $B_t(1),\ldots,B_t(u-1)$, and $A_t$ otherwise.

\item Apply the gate in Step 4.

\item Apply the inverse of the gates in Step 5.

\item Apply the gates in Step 3, Step 2, and Step 1 (in this order).
\end{enumerate}
After Stage $m$, we apply an $X^{g_m(0^m)}$ gate to $Y$. The proof of Lemma~\ref{exp} is 
as follows:

\begin{proof}
As with the proof of Lemma~\ref{correct}, a direct calculation shows that $R_m$ 
transforms the initial state of the output qubit into the state $|y\oplus g_m(s)\rangle$ and 
the output state of all the other qubits is the same as their initial state. Thus, 
$R_m$ computes $g_m$. Moreover, as with the proof of Lemma~\ref{complexity}, 
we can regard each stage of $R_m$ as an $O(m)$-depth circuit consisting of the gates in 
$\cal G$ and two PARITY$(m)$ gates. 

The remaining problem is to decompose the PARITY$(m)$ gate into the gates in $\cal G$. 
We construct an $O(m)$-depth quantum circuit for the PARITY$(m)$ gate using 
$m2^{m-1}$ uninitialized ancillary qubits. To describe the circuit, we use a 
parity gate. Here, for any integer $l\geq 1$, a parity gate on $l+1$ qubits implements the 
operation on $l+1$ qubits defined as
\begin{align*}
\left(\bigotimes_{j=1}^{l}|x_j\rangle\right)|y\rangle \mapsto 
\left(\bigotimes_{j=1}^{l}|x_j\rangle\right)|y \oplus \bigoplus_{j=1}^l x_j\rangle
\end{align*}
for any $y,x_j \in \{0,1\}$. The last input qubit is called the target qubit. Using a binary 
tree structure, we can regard this gate as an $O(\log l)$-depth circuit consisting 
of CNOT gates (although an $O(l)$-depth circuit is sufficient in the following argument). 
As in the construction of $R_m$, we have $m$ input qubits 
$S_1,\ldots,S_m$ and $2^m-1$ qubits, each of which is represented as $A_t$ in state 
$|a_t\rangle$. We also prepare $m2^{m-1}$ uninitialized ancillary 
qubits, which are divided into $m$ groups $D_1,\ldots,D_m$. Each $D_k$ consists of 
$2^{m-1}$ qubits, each of which is represented as $D_k(w)$ for any 
$w=w_1\cdots w_m \in \{0,1\}^m$ such that $w_k=1$, where the initial state of 
$D_k(w)$ is $|d_k(w)\rangle$ for any (unknown) $d_k(w) \in \{0,1\}$. 
As an example, when $m=2$, the circuit is depicted in Fig.~\ref{figure9}. The circuit 
for the PARITY$(m)$ gate is described as follows:
\begin{enumerate}
\item Apply a fan-out gate on $2^{m-1}+1$ qubits to $S_k$ and all $D_k(w)$'s with 
$w \in \{0,1\}^m$ such that $w_k=1$, 
for any $1 \leq k \leq m$ in parallel, where $S_k$ is the control qubit.

\item Apply a parity gate on $|t|+1$ qubits to $D_{k_1}(t),\ldots, D_{k_{|t|}}(t)$, and 
$A_t$ for any $t\in\{0,1\}^m \setminus \{0^m\}$ in parallel, where 
$A_t$ is the target qubit and $k_l$ is the position in $t$ such that $t_{k_l}=1$.

\item Apply the gates in Step 1 and Step 2 (in this order).
\end{enumerate}
Step 1 transforms the initial state of $D_k(w)$ into the state $|d_k(w)\oplus s_k\rangle$.  
Step 2 transforms the state of $A_t$ into the state
\begin{align*} 
|a_t \oplus \bigoplus_{k=1}^mt_ks_k \oplus \bigoplus_{l=1}^{|t|} d_{k_l}(t)\rangle.
\end{align*} 
Step 3 returns the state of $D_k(w)$ into its initial state and then transforms 
the state of $A_t$ into the state
\begin{align*}
|a_t \oplus \bigoplus_{k=1}^mt_ks_k \oplus 
\bigoplus_{l=1}^{|t|} d_{k_l}(t) \oplus \bigoplus_{l=1}^{|t|} d_{k_l}(t)\rangle 
=|a_t \oplus \bigoplus_{k=1}^mt_ks_k\rangle
\end{align*} 
as desired. A parity gate on $m+1$ qubits is decomposed into an $O(\log m)$-depth circuit. 
A fan-out gate on $2^{m-1}+1$ qubits is decomposed into an $O(m)$-depth circuit. 
Thus, the depth of the whole circuit is $O(m)$. This circuit for the PARITY$(m)$ gate allows 
us to regard $R_m$ as a circuit consisting of the gates in $\cal G$. The circuit $R_m$ uses 
$O(m2^m)$ uninitialized ancillary qubits and its depth is $O(m^2)$ since it has 
$m$ stages and the depth of each stage is $O(m)$.
\end{proof}

\begin{figure}
\centering
\includegraphics[scale=.32]{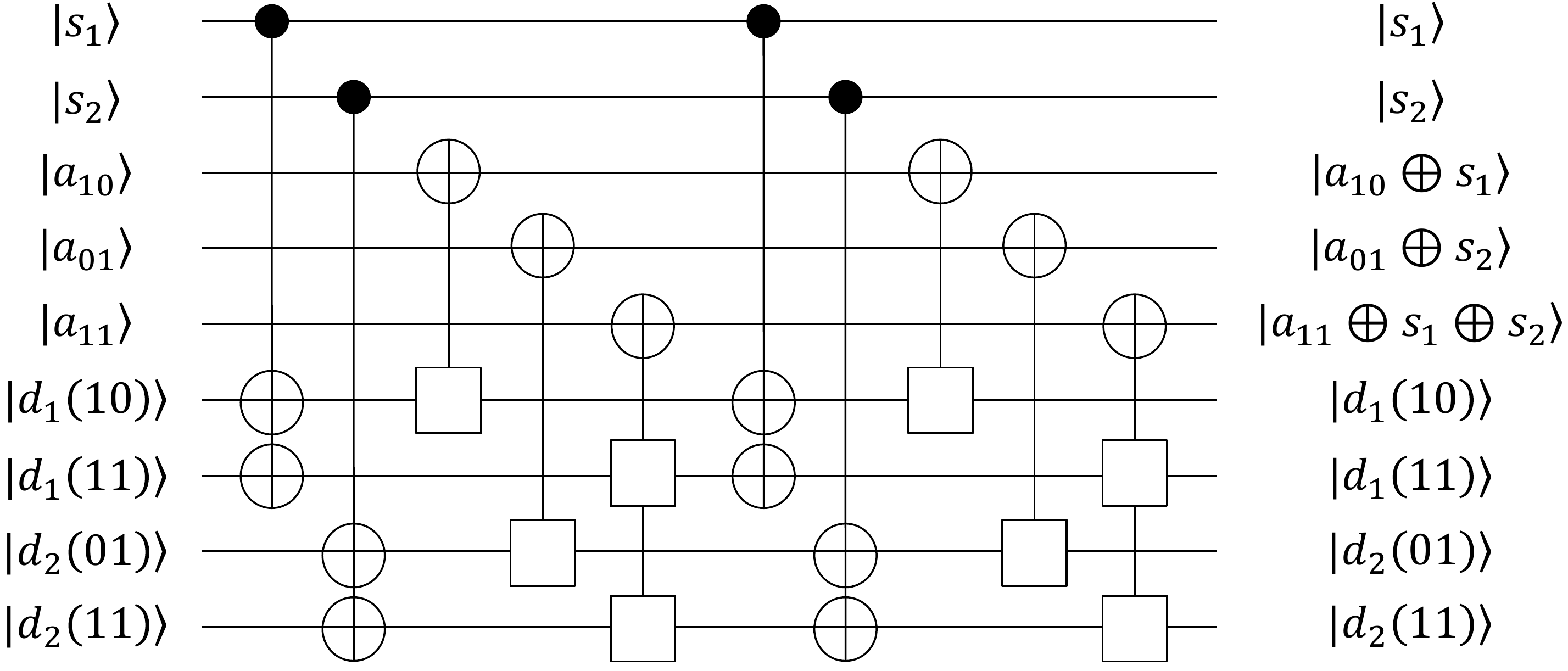}
\caption{Circuit for the PARITY$(2)$ gate. The gates other than the four 
fan-out gates are parity gates.}
\label{figure9}
\end{figure}

\subsection{Proof of Lemma~\ref{correct2}}\label{a-correct2}

\begin{proof}
As an example, $E_{2n}$ with $n=3$ is depicted in Fig.~\ref{figure6}. 
We assume that the initial states of $G(1),\ldots, G(n)$ are $|a_1\rangle,\ldots,|a_n\rangle$, 
respectively, for any (unknown) 
$a=a_1\cdots a_n \in \{0,1\}^n$. The states of $W_1,\ldots,W_n$ stay 
unchanged during the computation since the qubits are used only for control qubits. Thus, 
we consider the state of the remaining qubits $Y,G(1),\ldots,G(n),X_1,\ldots,X_n$. 
Steps 1--3 transform their initial state $|0\rangle |a\rangle |x\rangle$ into the state
\begin{align*}
\frac{1}{\sqrt{2}}|0\rangle|a_1\rangle\cdots|a_n\rangle(C_n|x\rangle) 
+ \frac{1}{\sqrt{2}}|1\rangle|a_1\oplus 1\rangle\cdots|a_n \oplus 1\rangle(C_n|x\rangle).
\end{align*}
Steps 4--5 transform this state into the state
\begin{align*}
& \frac{1}{\sqrt{2}}|0\rangle|a_1\rangle\cdots|a_n\rangle
\left\{C_n^\dag\left(\bigotimes_{j=1}^nZ_j^{a_jw_j}\right) C_n|x\rangle\right\} \\
& + \frac{1}{\sqrt{2}}|1\rangle|a_1\oplus 1\rangle\cdots|a_n \oplus 1\rangle 
\left\{C_n^\dag\left(\bigotimes_{j=1}^nZ_j^{(a_j\oplus 1)w_j}\right) C_n|x\rangle\right\}.
\end{align*}
Step 6 transforms this state into the state
\begin{align*}
& \frac{1}{2}|0\rangle|a\rangle
\left\{C_n^\dag\left(\bigotimes_{j=1}^nZ_j^{a_jw_j} + 
\bigotimes_{j=1}^nZ_j^{(a_j\oplus 1)w_j}\right) C_n|x\rangle\right\}\\
& + \frac{1}{2}|1\rangle|a\rangle
\left\{C_n^\dag\left(\bigotimes_{j=1}^nZ_j^{a_jw_j} - 
\bigotimes_{j=1}^nZ_j^{(a_j\oplus 1)w_j}\right) C_n|x\rangle\right\}.
\end{align*}
Thus, using the relationships $Z=Z^\dag$ and $Z^{a_jw_j+(a_j\oplus 1)w_j} = Z^{w_j}$ for any 
$a_j,w_j \in \{0,1\}$, we can compute the probability that the output of $E_{2n}$ is 0 as follows:
\begin{align*}
& \ \frac{1}{4}\langle x|C_n^\dag\left(\bigotimes_{j=1}^nZ_j^{a_jw_j} + 
\bigotimes_{j=1}^nZ_j^{(a_j\oplus 1)w_j}\right)^2 C_n|x\rangle \\
= & \ \frac{1}{2}\langle x|C_n^\dag\left(I + \bigotimes_{j=1}^nZ_j^{a_jw_j + (a_j\oplus 1)w_j} 
\right) C_n|x\rangle\\
= & \ \frac{1}{2}\langle x|C_n^\dag\left(I + \bigotimes_{j=1}^nZ_j^{w_j}\right) C_n|x\rangle
= \frac{1+ \langle x|C_n^\dag (\bigotimes_{j=1}^nZ_j^{w_j}) C_n|x\rangle}{2} 
= \frac{1+ F_n(x,w)}{2}.
\end{align*}
\end{proof}

\begin{figure}[t]
\centering
\includegraphics[scale=.32]{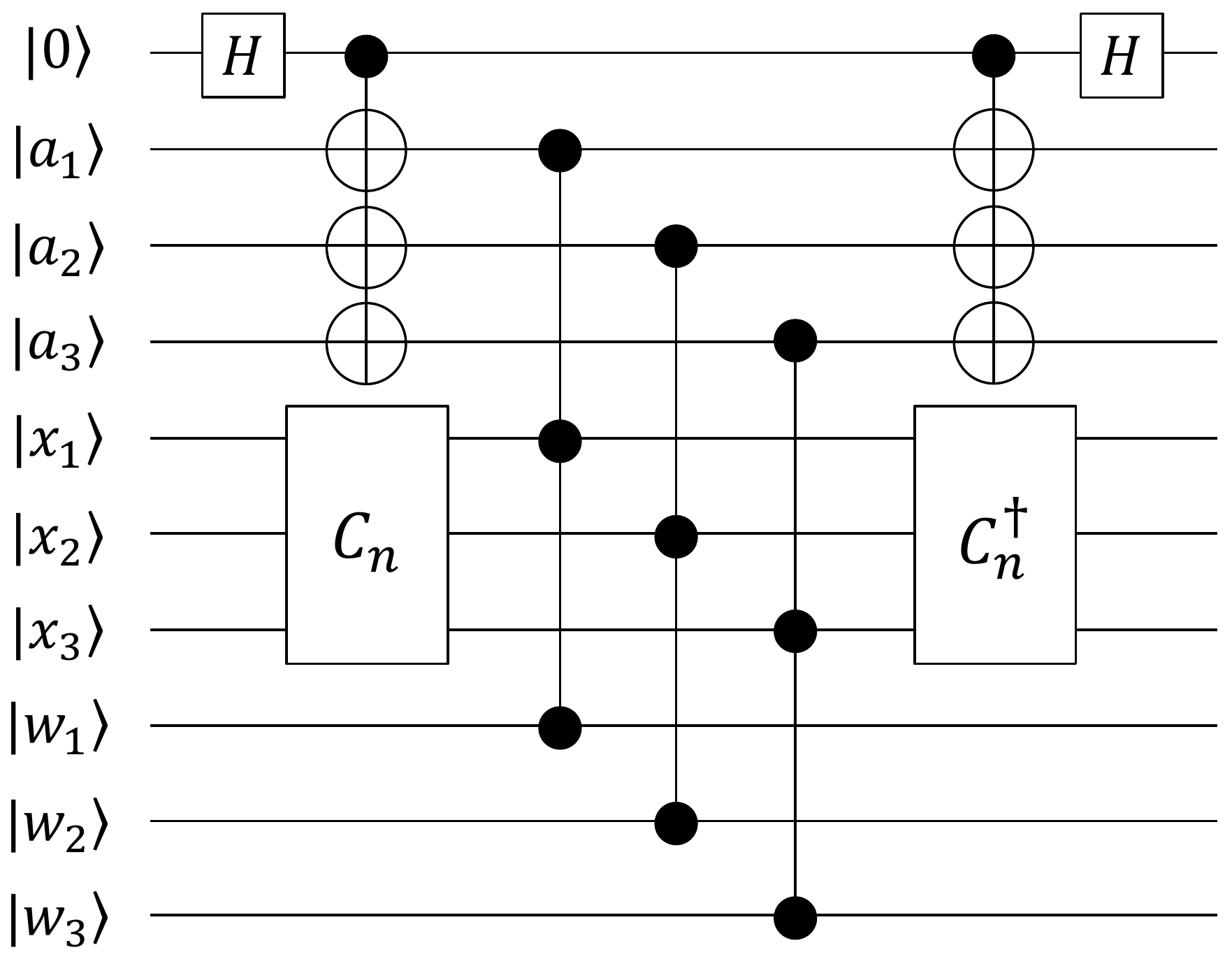}
\caption{Circuit $E_{2n}$ with $n=3$. The top qubit is the output qubit.}
\label{figure6}
\end{figure}

\subsection{Proof of Lemma~\ref{union}}\label{a-union}

\begin{proof}
Fix an arbitrary $\varepsilon >0$ and let $\varepsilon_1 = \varepsilon/k$. For any 
$x\in\{0,1\}^n$, $y\in\{0,1\}$, and $b\in\{0,1\}^{p+q}$, by the triangle 
inequality~\cite{Nielsen} and the definition of $V_l$, it holds that
\begin{align*}
\Delta(x,y,b) 
  = \ &
  ||C_n^{k+1}T_kV_k|x\circ y \circ b\rangle
- C_n^{k+1}C_n^k \cdots C_n^2C_n^1|x\circ y \circ b\rangle||\\
\leq \ &
||C_n^{k+1}T_kV_k|x\circ y \circ b\rangle 
-      C_n^{k+1}V_k|x\circ y \circ b\rangle||\\
& + ||C_n^{k+1}C_n^kT_{k-1}V_{k-1}|x\circ y \circ b\rangle
- C_n^{k+1}C_n^kV_{k-1}|x\circ y \circ b\rangle||\\
& + 
||C_n^{k+1}C_n^kC_n^{k-1}T_{k-2}V_{k-2}|x\circ y \circ b\rangle 
-      C_n^{k+1}C_n^kC_n^{k-1}V_{k-2}|x\circ y \circ b\rangle||\\
& \cdots \\
& + 
||C_n^{k+1}C_n^kC_n^{k-1}\cdots C_n^2 T_1 V_1|x\circ y \circ b\rangle 
-      C_n^{k+1}C_n^kC_n^{k-1}\cdots C_n^2 V_1|x\circ y \circ b\rangle||\\
 = \ & 
\sum_{l=1}^k\Delta_l(x,y,b).
\end{align*}
This implies that, for any $y\in\{0,1\}$ and $b\in\{0,1\}^{p+q}$,
\begin{equation}\label{eq2}
{\rm Pr}[\Delta(U_n,y,b) < k\varepsilon_1] \geq 
{\rm Pr}[\bigwedge_{l=1}^k\Delta_l(U_n,y,b) < \varepsilon_1].
\end{equation}
By the union bound and the Markov's inequality~\cite{Goldreich},
\begin{align*}
{\rm Pr}[\bigvee_{l=1}^k\Delta_l(U_n,y,b) \geq \varepsilon_1] & \leq 
\sum_{l=1}^k{\rm Pr}[\Delta_l(U_n,y,b) \geq \varepsilon_1] =
\sum_{l=1}^k {\rm Pr}[\Delta_l(U_n,y,b)^2 \geq \varepsilon_1^2] \\ 
& \leq
\frac{1}{\varepsilon_1^2}\sum_{l=1}^k{\rm E}[\Delta_l(U_n,y,b)^2],
\end{align*}
which implies that
\begin{align*}
{\rm Pr}[\bigwedge_{l=1}^k\Delta_l(U_n,y,b) < \varepsilon_1] \geq 
1-\frac{1}{\varepsilon_1^2}\sum_{l=1}^k{\rm E}[\Delta_l(U_n,y,b)^2].
\end{align*}
This relationship with (\ref{eq2}) implies that
\begin{align*}
{\rm Pr}[\Delta(U_n,y,b) < k\varepsilon_1] \geq 
1-\frac{1}{\varepsilon_1^2}\sum_{l=1}^k{\rm E}[\Delta_l(U_n,y,b)^2],
\end{align*}
which immediately implies the desired relationship since $\varepsilon_1 = \varepsilon/k$.
\end{proof}

\subsection{Proof of Lemma~\ref{bunsu}}\label{a-bunsu}

\begin{proof}
First, we show that, for any $1 \leq l \leq k$, $y\in\{0,1\}$, and $b\in\{0,1\}^{p+q}$,
\begin{align*}
{\rm E}[\Delta_l(U_n,y,b)^2] = \frac{4}{2^n} \sum_{j\in\{0,1\}^{u_l}} \sum_{x\in\{0,1\}^n}
|g_{x\circ y\circ b}^{(l)}(1^{t_l} \circ j)|^2.
\end{align*}
This can be shown by the following simple calculation. 
For any $1 \leq l \leq k$, $x\in\{0,1\}^n$, $y\in\{0,1\}$, and $b\in\{0,1\}^{p+q}$,
\begin{align*}
T_lV_l|x\circ y \circ b\rangle - V_l|x\circ y \circ b\rangle
= -2 \sum_{j\in\{0,1\}^{u_l}} g_{x\circ y\circ b}^{(l)}(1^{t_l} \circ j)|1^{t_l}\circ j\rangle.
\end{align*}
This implies that
\begin{align*}
\Delta_l(x,y,b)^2 = 4\sum_{j\in\{0,1\}^{u_l}} 
|g_{x\circ y\circ b}^{(l)}(1^{t_l} \circ j)|^2.
\end{align*}
Thus,
\begin{align*}
{\rm E}[\Delta_l(U_n,y,b)^2] = \frac{1}{2^n}\sum_{x\in\{0,1\}^n} \Delta_l(x,y,b)^2
= \frac{4}{2^n}\sum_{j\in\{0,1\}^{u_l}} \sum_{x\in\{0,1\}^n}
|g_{x\circ y\circ b}^{(l)}(1^{t_l} \circ j)|^2,
\end{align*}
which is the desired relationship.

Then, we show that, for any $1 \leq l \leq k$, $i\in\{0,1\}^{t_l}$, 
and $j\in\{0,1\}^{u_l}$,
\begin{align*}
\sum_{x\in\{0,1\}^n}\sum_{y\in\{0,1\}}\sum_{b \in\{0,1\}^{p+q}} 
|g_{x\circ y\circ b}^{(l)}(i \circ j)|^2 = 1.
\end{align*}
This can also be shown by the following simple calculation. For any $1 \leq l \leq k$, the 
unitary operation $V_l$ can be represented as
\begin{align*}
V_l & =  \sum_{x\in\{0,1\}^n}\sum_{y\in\{0,1\}}\sum_{b \in\{0,1\}^{p+q}}
V_l|x\circ y\circ b\rangle \langle x \circ y \circ b|
\\
& =  \sum_{x\in\{0,1\}^n}\sum_{y\in\{0,1\}}\sum_{b \in\{0,1\}^{p+q}}
\sum_{i\in\{0,1\}^{t_l}}\sum_{j\in\{0,1\}^{u_l}} g_{x\circ y\circ b}^{(l)}(i \circ j)|i\circ j\rangle
\langle x \circ y \circ b|.
\end{align*}
This implies that, for any $i\in\{0,1\}^{t_l}$ and $j\in\{0,1\}^{u_l}$,
\begin{align*}
V_l^\dag |i\circ j\rangle = \sum_{x\in\{0,1\}^n}\sum_{y\in\{0,1\}}\sum_{b \in\{0,1\}^{p+q}}
 g_{x\circ y\circ b}^{(l)}(i \circ j)^*
|x \circ y \circ b\rangle.
\end{align*}
A direct calculation shows that
\begin{align*}
1 = \langle i\circ j| V_l V_l^\dag |i\circ j\rangle = 
\sum_{x\in\{0,1\}^n}\sum_{y\in\{0,1\}}\sum_{b \in\{0,1\}^{p+q}} 
|g_{x\circ y\circ b}^{(l)}(i \circ j)|^2,
\end{align*}
which is the desired relationship.

The above relationships imply Lemma~\ref{bunsu} as follows. For any $y\in\{0,1\}$ 
and $b\in\{0,1\}^{p+q}$,
\begin{align*}
\sum_{l=1}^k {\rm E}[\Delta_l(U_n,y,b)^2] & = \frac{4}{2^n} \sum_{l=1}^k 
\sum_{j\in\{0,1\}^{u_l}} \sum_{x\in\{0,1\}^n}
|g_{x\circ y\circ b}^{(l)}(1^{t_l} \circ j)|^2\\
& \leq \frac{4}{2^n}\sum_{l=1}^k\sum_{y\in\{0,1\}}\sum_{b \in\{0,1\}^{p+q}} 
\sum_{j\in\{0,1\}^{u_l}}\sum_{x\in\{0,1\}^n}|g_{x\circ y\circ b}^{(l)}(1^{t_l} \circ j)|^2 \\
& =  \frac{4}{2^n}\sum_{l=1}^k \sum_{j\in\{0,1\}^{u_l}}1 =\frac{4}{2^n}
 \sum_{l=1}^k 2^{u_l} = \sum_{l=1}^k  \frac{2^{p+q+3}}{2^{t_l}}
 \leq \frac{k2^{p+q+3}}{2^{t_{\rm min}}},
 \end{align*}
which is the desired first relationship. In particular,
\begin{align*}
\frac{4}{2^n}\sum_{l=1}^k\sum_{b \in\{0,1\}^{p+q}} 
\sum_{j\in\{0,1\}^{u_l}}\sum_{x\in\{0,1\}^n}|g_{x\circ 0\circ b}^{(l)}(1^{t_l} \circ j)|^2
 \leq \frac{k2^{p+q+3}}{2^{t_{\rm min}}}.
\end{align*}
We consider the value
\begin{equation}\label{minimize}
\sum_{l=1}^k {\rm E}[\Delta_l(U_n,0,0^p\circ a)^2] =\frac{4}{2^n}\sum_{l=1}^k 
\sum_{j\in\{0,1\}^{u_l}}\sum_{x\in\{0,1\}^n}|g_{x\circ 0\circ 0^p\circ a}^{(l)}(1^{t_l} \circ j)|^2
\end{equation}
for any $a \in \{0,1\}^q$. There exists some $a'\in\{0,1\}^q$ such that it minimizes this 
value, i.e., value (\ref{minimize}) with $a'$ is less than or equal to that 
with any other $a \in \{0,1\}^q$. It holds that
\begin{align*}
2^q\sum_{l=1}^k {\rm E}[\Delta_l(U_n,0,0^p\circ a')^2] \leq 
\sum_{a \in \{0,1\}^q}\sum_{l=1}^k{\rm E}[\Delta_l(U_n,0,0^p\circ a)^2]
\leq \frac{k2^{p+q+3}}{2^{t_{\rm min}}},
\end{align*}
which immediately implies the desired second relationship.
\end{proof}

\end{document}